\documentclass[conference,compsoc]{IEEEtran}
\IEEEoverridecommandlockouts

\usepackage{nicefrac}
\usepackage{cite}
\usepackage{amsmath,amssymb,amsfonts}
\usepackage[ruled,vlined]{algorithm2e}
\usepackage{algpseudocode}
\usepackage{graphicx}
\usepackage[utf8]{inputenc} 
\usepackage{subcaption}
\usepackage{textcomp}
\usepackage{xcolor}
\usepackage{bbm}
\usepackage{mathtools}
\usepackage{diagbox}
\usepackage{array}
\usepackage{amsthm}

\usepackage{breqn}
\usepackage{enumitem}
\usepackage{physics}
\usepackage{thm-restate}
\usepackage{url}
\usepackage{tikz-cd}
\usepackage{multirow}
\usepackage[framemethod=tikz]{mdframed}
\theoremstyle{definition}
\newtheorem{definition}{Definition}[section]
\theoremstyle{remark}
\newtheorem{remark}{Remark}

\newtheorem{example}{Example}[section]

\DeclareMathOperator*{\argmin}{arg\,min}
\usepackage{comment}
\usepackage{wrapfig}

\definecolor{mycolor}{rgb}{0.122, 0.435, 0.698}

\newmdenv[innerlinewidth=0.5pt, roundcorner=4pt,linecolor=mycolor,innerleftmargin=6pt,
innerrightmargin=6pt,innertopmargin=6pt,innerbottommargin=6pt]{mybox}

\usepackage{enumitem}

\allowdisplaybreaks
\newcommand\revision[1]{\textcolor{black}{#1}}

\def\BibTeX{{\rm B\kern-.05em{\sc i\kern-.025em b}\kern-.08em
    T\kern-.1667em\lower.7ex\hbox{E}\kern-.125emX}}

\begin{document}
\title{PRIVIC: A privacy-preserving method for incremental collection of location data
\thanks{Supported by the European Research Council (ERC) project HYPATIA under the European Union’s Horizon 2020 research and innovation programme (grant agreement no. 835294).}
}
\author{\IEEEauthorblockN{1\textsuperscript{st} Sayan Biswas}
\IEEEauthorblockA{\textit{INRIA, LIX, \'Ecole Polytechnique}\\
Palaiseau, France \\
sayan.biswas@inria.fr}
\and
\IEEEauthorblockN{2\textsuperscript{nd} Catuscia Palamidessi}
\IEEEauthorblockA{\textit{INRIA, LIX, \'Ecole Polytechnique}\\
Palaiseau, France \\
catuscia@lix.polytechnique.fr}
}
\maketitle

\begin{abstract}
 With recent advancements in technology, the threats of privacy violations of individuals' sensitive data are surging. Location data, in particular, have been shown to carry a substantial amount of sensitive information. A standard method to mitigate the privacy risks for location data consists in adding noise to the true values to achieve geo-indistinguishability (geo-ind). However, geo-ind alone is not sufficient to cover all privacy concerns. In particular, isolated locations are not sufficiently protected by the state-of-the-art Laplace mechanism (LAP) for geo-ind. In this paper, we focus on a mechanism based on the Blahut-Arimoto algorithm (BA) from the rate-distortion theory. We show that BA, in addition to providing geo-ind, enforces an elastic metric that mitigates the problem of isolation. Furthermore, BA provides an optimal trade-off between information leakage and quality of service. We then proceed to study the utility of BA in terms of the statistics that can be derived from the reported data, focusing on the inference of the original distribution. To this purpose, we de-noise the reported data by applying the iterative Bayesian update (IBU), an instance of the expectation-maximization method. It turns out that BA and IBU are dual to each other, and as a result, they work well together, in the sense that the statistical utility of BA is quite good and better than LAP for high privacy levels. Exploiting these properties of BA and IBU, we propose an iterative method, PRIVIC, for a privacy-friendly incremental collection of location data from users by service providers. We illustrate the soundness and functionality of our method both analytically and with experiments. 
\end{abstract}


\section{Introduction}

As the need and development of various kinds of research and analysis using personal data are becoming more and more significant, the risk of privacy violations of sensitive information of the data owners is also increasing manifold. One of the most successful proposals to address the issue of privacy protection is \emph{differential privacy (DP)}~\cite{DworkDP1,DworkDP2}, a mathematical property that makes it difficult for an attacker to detect the presence of a record in a dataset. This is typically achieved by answering queries performed on the dataset in a (controlled) noisy fashion. Lately, the \emph{local variant of differential privacy (LDP)}~\cite{DuchiLDP} has gained popularity due to the fact that the noise is applied at the data owner's end without needing a trusted curator. LDP is particularly suitable for situations where a data owner is a user who communicates her personal data in exchange for some service. One such scenario is the use of location-based services (LBS), where a user typically sends her location in order to obtain  information like the shortest path to a destination, nearby points of interest, traffic information, etc. The security and the convenience of  implementing the local model directly on a user's device (tablets, smartphones, etc.) make LDP very appealing. 

Typically, in exchange for their service,  providers incrementally collect  their users' data and then make them available to other parties which process them to provide useful statistics to companies and institutions. Obviously, the statistical precision of the collected data is essential for the quality of the analytics performed (\emph{statistical utility}). However, injecting noise locally into the data to protect the privacy of the users usually has a negative effect on the statistical utility. Additionally, the noise degrades the \emph{quality of service} (QoS) as well, since, obviously, the service results from the elaboration of the information received. 

Substantial research has been done to address the privacy-utility trade-off in the context of DP. In LDP, the primary focus has been to optimize the utility from the data collector's perspective, i.e., devising mechanisms and post-processing methods that would allow deriving the most accurate statistics from the collection of the noisy data~\cite{DuchiLDP,GoogleRappor}. In contrast, in domains such as location privacy, the focus usually has been on optimizing the QoS, i.e., the utility from the point of view of the users. In particular, this is the case for the framework proposed by Shokri et al.~\cite{ShokriLocationPrivacy, ShokriPrivacyGames}. 

We argue that it is important to meet the interest of all parties involved, and hence to consider both kinds of utility at the same time. Hence, the first goal of this paper is to develop a \emph{location-privacy preserving mechanism (LPPM)} that, in addition to providing formal location-privacy guarantees,  preserves as much as possible \emph{both}  the statistical utility and  the QoS.

Now, one may think that statistical utility and QoS are aligned since they both benefit from preserving as much original information as possible under the privacy constraint. However, this is not true in general:  the optimization of statistical utility does not necessarily imply a significant improvement in the QoS, nor vice-versa. A counterexample is provided by Example~\ref{examp:BadLPPM} in Section~\ref{sec:PRIVIC}. Hence, 
the preservation of both statistical utility and  QoS is trickier than it may appear at first sight.

One of the approaches which have been proposed to protect location privacy is  \emph{geo-indistinguishability} (geo-ind)~\cite{AndresKostasCatuscia_GeoInd},  which essentially obfuscates locations based on the distance between them. This idea works particularly well for protecting the precision of the location as it ensures that an attacker would not be able to differentiate between points that are close on the map by observing the reported noisy location. At the same time, it does not inject an enormous amount of noise that would be necessary to make far-away locations indistinguishable. Moreover, geo-ind has been shown to formally satisfy the basic sequential compositionality theorem~\cite{galli2022group}, just like DP and its local variant. Although this approach of distance-based obfuscation seems enticing at a first glance, one of the issues it poses is that it may leave the geo-spatially isolated locations vulnerable, i.e., identifiable despite being formally geo-indistinguishable~\cite{Elastic_Kostas}. To improve the situation,  \cite{Elastic_Kostas} introduced the notion of \emph{elastic distinguishability metrics}, which essentially leads to injecting more noise  when the location to protect is isolated.

The \emph{Blahut-Arimoto algorithm} (BA)~\cite{Blahut72computationof, Arimoto1972AnAF} from rate-distortion theory (a branch of information theory) 
Pareto-optimizes the trade-off between \emph{mutual information} (MI) and average distortion. This property is appealing in the context of privacy because MI is often considered a measure of information leakage and average distortion is
a commonly used metric for quantifying QoS. Moreover, BA was proven to satisfy geo-ind in \cite{Oya:17:CCS} opening a door to study it as a potential LPPM. In this paper, we start off by exploring the privacy-preserving properties of BA and comparing them with those of the \emph{Laplace mechanism} (LAP)~\cite{AndresKostasCatuscia_GeoInd} which is considered as the state-of-the-art mechanism for geo-ind. We show that, besides geo-ind, BA provides an elastic distinguishability metric and, hence, protects even the most isolated points in the map, unlike LAP. We then examine the statistical utility, focusing on the estimation of the most general statistical information, namely the distribution of the original location data (true distribution).  The ``best'' estimation is known in statistics as the \emph{maximum likelihood estimation} (MLE), and can be 
 computed  using the \emph{iterative Bayesian update} (IBU)~\cite{AgarwalIBU}, an instance of the \emph{expectation maximization} (EM) method.
 We discover a duality between BA and IBU, which in our opinion is quite intriguing, because BA and IBU were developed in different contexts, using different concepts and metrics, and for completely different purposes. 
 We prove experimentally that the statistical utility of BA is very good, i.e., the MLE is very close to the true distribution. We conjecture that this is probably due to the duality between the mechanism that injects the noise (BA) and the one that de-noises the noisy data (IBU). In any case, the experiments show that the statistical utility of BA
 outperforms that of LAP for high levels of privacy, eventually becoming comparable as the level of privacy decreases. 

One important point to note is that  BA requires the knowledge of the original distribution to provide the optimal mechanism. When it is fed with only an approximation of the distribution, it only provides an approximated result. 
We acknowledge that the distribution of the original data is usually off-limits and, even when available, it typically gets outdated over time. In any case, we can soundly assume that it is not available because it is essentially the reason for collecting the data. Hence we have a vicious circle: we want to collect data in a privacy-friendly fashion to estimate the original distribution while wanting to use a privacy mechanism that requires knowing a good approximation of the original distribution. Motivated by this dilemma, we propose PRIVIC, an incremental data collection method providing extensive privacy protection for the users of LBS's, while retaining a high utility for both them and the service providers, and ensuring that both parties, acting in their best interest, would benefit from the end mechanism.

\revision{Finally, we prove formally the convergence of PRIVIC to the true distribution, and illustrate empirically the privacy-utility trade-off of our method.} The experiments also demonstrate the efficacy of combining BA and IBU, in that the estimation of the original distribution is very accurate, especially when measured using a notion of distance between distributions compatible with the ground distance used to measure the QoS (e.g., the Earth Mover's distance). All the experiments were performed using real location data from the Gowalla dataset for Paris and San Francisco. 

\vspace{-1.6mm}
\subsubsection*{Contributions}
The key contributions of this paper are:

\begin{enumerate}
\vspace{-0.9mm}
\item We show, analytically and with experiments on real datasets, that the BA mechanism, in addition to geo-ind, provides an elastic distinguishability metric. As such, it protects the privacy of isolated locations, which the standard LAP for geo-ind fails at. 

\item  We prove that BA produces an invertible mechanism, which means that the MLE is unique. This is crucial to prove that the IBU always converges to the true distribution and that, therefore, we can get a good statistical utility.

\item We establish a duality between BA and IBU, thus demonstrating a connection between rate-distortion theory and the expectation-maximization method from statistics.

\item We show experimentally that BA provides a better statistical utility than LAP for high levels of privacy, eventually becoming comparable as the level of privacy decreases.

\item Since the construction of the optimal BA requires precise knowledge of the true distribution, we propose an iterative method (PRIVIC) that  alternates between BA and IBU, thus  getting a better and better estimation of the true distribution  as more (noisy) data get collected. We show, both formally and with experiments on real location datasets, that PRIVIC converges to the true distribution. In summary, PRIVIC produces a geo-indistinguishable LPPM with an elastic distinguishability metric, which optimizes the trade-off with the QoS and provides high statistical utility.

\item \revision{We investigate the effect on the privacy guarantees of our method by considering adversarial users who report their locations falsely to compromise the privacy of the isolated locations in the map.}
 
\end{enumerate}
 \vspace{0.9mm}

\subsubsection*{Related Work}~\label{sec:related}
The trade-off between privacy and utility has been widely studied in the literature~\cite{BrickellPvcyUtility:2008, TraftTradeoff:2014}. Optimization techniques for DP and utility for statistical databases have been analyzed by the community from various perspectives~\cite{GhoshOptimalPrivacyUtility:2012, GupteOptimalPrivacyforMinMaxAgents, LiOptimalLinearQueries:2010}. There have been works focusing on devising privacy mechanisms that are optimal to limit the privacy risk against Bayesian inference attacks while maximizing the utility~\cite{ShokriLocationPrivacy, ShokriPrivacyGames}. In \cite{Oya:17:CCS}, Oya et al. examine an optimal LPPM w.r.t. various privacy and utility metrics for the user.

In \cite{Oya:19:EuroSnP}, Oya et al. consider the optimal LPPM proposed by Shokri et al. in \cite{ShokriLocationPrivacy} which maximizes a notion of privacy (the \emph{adversarial error}) under some bound on the QoS. The construction of the optimal LPPM requires the knowledge of the original distribution, and \cite{Oya:19:EuroSnP} uses the EM method to estimate it and design \emph{blank-slate models} empirically shown to outperform the traditional hardwired models. However, a  problem with their approach is that there may exist LPPMs that are optimal in the sense of~\cite{ShokriLocationPrivacy}, but with no statistical utility, see Example \ref{examp:BadLPPM} in Section~\ref{sec:PRIVIC}. Furthermore, for the mechanisms considered in \cite{Oya:19:EuroSnP} the EM method may fail to converge to the true distribution.  Indeed, \cite{EhabConvergenceIBU} points out various mistakes in the results of \cite{AgarwalIBU}, on which \cite{Oya:19:EuroSnP} intrinsically relies to prove the convergence of their method.

\cite{Romanelli:20:CSF} proposed a method for generating privacy mechanisms that tend to minimize mutual information using an ML-based approach. However, this work assumes the knowledge of the exact prior from the beginning, unlike ours. Moreover, \cite{Romanelli:20:CSF} does not provide formal guarantees for location privacy (e.g., geo-ind) which is one of the main aspects captured by our work. In \cite{Zhang_trace}, Zhang et al. consider the Blahut-Arimoto algorithm in the context of location privacy. However, their proposed method also requires the knowledge of the prior distribution to construct the LPPM. Additionally, \cite{Zhang_trace} focuses on measuring privacy for the trace of a single user. On the contrary, our notion of privacy assumes the collection of single check-ins (or check-ins separated in time) by a set of users.

The Laplace mechanism has been rigorously studied in the literature in various scenarios as the cutting-edge standard to achieve geo-ind~\cite{AndresKostasCatuscia_GeoInd, galli2022group, atmaca2022privacy} and has been proven to be optimal for one-dimensional data w.r.t. Bayesian utility~\cite{Fernandes:21:LICS}. Despite its wide popularity, it has been recently criticized due to its limitation to protect geo-spatially isolated points from being identified by adversaries~\cite{Elastic_Kostas}. The authors of \cite{Elastic_Kostas} addressed this concern by proposing the idea of \emph{elastic distinguishability metrics}.

Our paper also considers mutual information (MI) as an additional privacy guarantee. MI and its closely related variants (e.g. conditional entropy) have been shown to nurture a compatible relationship with DP~\cite{Cuff:16:CCS}. MI measures the correlation between observations and secrets, and its use as a privacy metric is widespread in the literature. Some key examples are: gauging anonymity~\cite{Zhu:05:ICDCS, Chatzikokolakis:08:IC}, estimating privacy in training ML models with a typical cross-entropy loss function~\cite{Abadi:16:CoRR, Tripathy:19:ACCC, Romanelli:20:CSF, Huang:17:Entropy}, and assessing location-privacy~\cite{Oya:17:CCS}.

A popular choice of utility metric for the users is the \emph{average distortion}, which quantifies the expected quality loss of the service due to the noise induced by the  mechanism. Such a metric has gained the spotlight in the community~\cite{AndresKostasCatuscia_GeoInd, NicolasKostasCatusciaOptimalGeoInd, ChatzikokolakisPalamidessiStronati2015, ChatzikokolakisElSalamounyPalamidessi_PracticalLocation2017, ShokriLocationPrivacy} due to its intuitive and simple nature. On the other hand, a standard notion  of statistical utility for the data consumer is the precision of the estimation of  the distribution on the original data from that of the noisy data. Iterative Bayesian update~\cite{AgarwalIBU,agrawal2005privacy} provides one of the most flexible and powerful estimation techniques  and has recently become in the focus of the community ~\cite{EhabConvergenceIBU, EhabGIBU}.

Incremental and privacy-friendly data collection has been explored both in the context of $k$-anonymity~\cite{Byun_incremental_data_dissemination, Byun_Secure_Anon_incremental,anjum2017tau} and DP~\cite{wang2016using, Gursoy_utilityawareDC}. However, to the best of our knowledge, the problem of providing a rather robust privacy guarantee while preserving utility for both data owners and data consumers has not been addressed by the community so far. 

\subsubsection*{Plan of the paper}
Section~\ref{sec:prelims} introduces preliminary ideas from the literature relevant to this work. Section~\ref{sec:LPPM_with_BA} highlights BA as an LPPM because of its extensive privacy-preserving properties. Section~\ref{sec:duality} establishes the duality between BA and IBU. Section~\ref{sec:PRIVIC} explains our proposed method (PRIVIC). Section~\ref{sec:experiments} exhibits the working of PRIVIC with experiments using real locations from the Gowalla dataset illustrating the convergence of our method. \revision{Section~\ref{sec:vul_PRIVIC} discusses and illustrates with experiments the vulnerability of PRIVIC under adversarial data submission} and Section~\ref{sec:conclusion} concludes. Appendices \ref{app:proofs} and \ref{app:tables} contain the proofs of the theorems derived in the paper and the relevant tables supporting the experimental analysis of PRIVIC, respectively.


\section{Preliminaries}~\label{sec:prelims}
\vspace{-0.5cm}
\subsection{Standards of privacy}
\begin{definition}[$d$-\emph{privacy}, a.k.a. \emph{ metric privacy}~\cite{chatzikokolakis_dprivacy}]
For any space $\mathcal{X}$ equipped with a metric $d:\mathcal{X}^2\mapsto \mathbb{R}_{\geq 0}$ and an output space $\mathcal{Y}$, a mechanism $\mathcal{R}:\mathcal{X}\mapsto \mathcal{Y}$ is $\epsilon$-\emph{$d$-private} if $\mathbb{P}[\mathcal{R}(x)=y] \leq e^{\epsilon d(x,x')}\,\mathbb{P}[\mathcal{R}(x')= y]$ for every $x,x'\in\,\mathcal{X}$ and $y\in\mathcal{Y}$.
\end{definition}
Note  that: 
\begin{itemize}
 \item 
Setting $d$ as the \emph{discrete metric} on any $\mathcal{X}$, we obtain the definition of \emph{local differential privacy (LDP)}~\cite{DuchiLDP}.
\item 
Setting $\mathcal{X}=\mathcal{Y}=\mathbb{R}^2$ and $d$ as the \emph{Euclidean metric}, we get the definition of \emph{geo-ind}~\cite{AndresKostasCatuscia_GeoInd}.
\end{itemize}

\begin{definition}[Mutual information\cite{ShannonInfoTheory}]
\label{def:MI}
Let $(X,Y)$ be a pair of random variables defined over the discrete space $\mathcal{X}\times\mathcal{Y}$ such that $\mu$ is the joint \emph{probability mass function} (PMF) of $X$ and $Y$, and $p_{X}$ and $p_{Y}$ are the marginal PMFs of $X$ and $Y$, respectively, and $p_{X|Y}$ is the conditional probability of $X$ given $Y$. Then the (Shannon) \emph{entropy} of $X$, $H(X)$, is defined as $H(X)= -\sum\limits_{x\in\mathcal{X}}p_X(x) \log p_X(x)$. The \emph{residual  entropy} of $X$ given $Y$ is defined as $H(X|Y)= \sum\limits_{y\in\mathcal{Y}}p_Y(y) H(X|Y=y) =  -\sum\limits_{y\in\mathcal{Y}}p_Y(y)\sum\limits_{x\in\mathcal{X}}p_{X|Y}p(x|y) \log p_{X|Y}(x|y)$, and, finally, the 
\emph{mutual information (MI)} is given by: $$I(X|Y)=H(X) - H(X|Y) =\sum\limits_{x\in\mathcal{X}}\sum\limits_{y\in\mathcal{Y}}\mu(x,y)\log\frac{\mu(x,y)}{p_{X}(x)p_{Y}(y)}$$

\end{definition}

\begin{remark}
    MI has often been used as a notion of privacy (and security) in the literature. In particular, 
    \cite{Kopf:07:CCS} has provided an operational interpretation of MI in terms of an attacker model. On the other hand, other researchers have strongly criticized the use of Shannon entropy and MI as  measures of privacy, see for example \cite{Syverson:13:SecProto}. 

    We do not take sides in this controversy: for us, MI is only a means to construct a mechanism that provides geo-ind under an elastic metric, which is our reference  privacy notion.  
\end{remark}

\subsection{Notions of utility}

\begin{definition}[Quality of service]~\label{def:AvgDist}
For discrete spaces $\mathcal{X}$ and $\mathcal{Y}$, let $d\colon\mathcal{X}\times\mathcal{Y}\rightarrow \mathbb{R}_{\geq 0}$ be any distortion metric (a generalization of the notion of distance). Let $X$ be a random variable on $\mathcal{X}$ with PMF $p_{\mathcal{X}}$ and $\mathcal{C}$ be any randomizing mechanism where $\mathcal{C}_{xy}$ is the probability of $x$ being mapped by $\mathcal{C}$ into $y$. We define the \emph{quality of service (QoS)} of $X$ for $\mathcal{C}$ as the \emph{average distortion w.r.t. $d$}, given as:
\begin{equation*}
    AvgD(X,\mathcal{C},d)=\sum\limits_{x\in\mathcal{X}}\sum\limits_{y\in\mathcal{Y}}p_{\mathcal{X}}(x)\mathcal{C}_{xy}d(x,y)
\end{equation*}
\end{definition}

\begin{definition}[Full-support probability distribution]\label{def:fullsupportPMF} Let $\theta$ be a probability distribution defined on the space $\mathcal{X}$. $\theta$ is a \emph{full-support} distribution on $\mathcal{X}$ if $\theta(x)>0$ for every $x\in\mathcal{X}$.

\end{definition}

\begin{definition}[Iterative Bayesian update~\cite{AgarwalIBU}]\label{def:IBU}
Let $\mathcal{C}$ be a privacy mechanism that locally obfuscates points from the discrete space $\mathcal{X}$ to $\mathcal{Y}$ such that $\mathcal{C}_{xy}=\mathbb{P}[y|x]$ for all $x,y\in \mathcal{X},\mathcal{Y}$. Let $X_1,\ldots,X_n$ be i.i.d. random variables on $\mathcal{X}$ following some PMF $\pi_{\mathcal{X}}$. Let $Y_i$ denote the random variable of the output when $X_i$ is obfuscated with $\mathcal{C}$. 

Let $\overline{y}=\{y_1,\ldots,y_n\}$ be a realisation of $\{Y_1,\ldots, Y_n\}$ and $\vb*{q}$ be the empirical distribution obtained by counting the frequencies of each $y$ in $\overline{y}$. The \emph{iterative Bayesian update (IBU)} 
estimates $\pi_{\mathcal{X}}$ by converging to its maximum likelihood estimate (MLE) with the knowledge of $\vb*{q}$ and  $\mathcal{C}$. IBU works as follows:
\begin{enumerate}
    \item Start with any full-support PMF $\theta_0$ on $\mathcal{X}$.
    \item Iterate $ \theta_{r+1}(x)=\sum\limits_{y\in\mathcal{Y}}\vb*{q}(y)\frac{\theta_r(x)\mathcal{C}_{xy}}{\sum\limits_{z\in\mathcal{X}}\theta_r(z)\mathcal{C}_{zy}}$ for all $x\in\mathcal{X}$.
\end{enumerate}
\end{definition}

The convergence of IBU has been studied in \cite{AgarwalIBU,EhabConvergenceIBU}. 
For a given set of observed locations, the limiting estimate  $\hat{\pi}_{\mathcal{X}}=\lim\limits_{r\to\infty}\theta_{r}$  is well-defined by the privacy mechanism in use, $\mathcal{C}$, and the empirical distribution of the noisy locations, $\vb*{q}$. We will functionally denote $\hat{\pi}_{\mathcal{X}}$ as $\texttt{IBU}(\vb*{q},\mathcal{C})$. 

\color{black}

Next, we recall a generalization of IBU from the literature that we use in this work. Generalized IBU (GIBU)~\cite{EhabGIBU} applies IBU in parallel to several empirical distributions derived from the application of (possibly different) obfuscation mechanisms to various sets of samples from the same distribution.

\begin{definition}[Generalized iterative Bayesian update~\cite{EhabGIBU}]~\label{def:GIBU}

Let $\vb*{x}^{(1)},\ldots,\vb*{x}^{(N)}$, with $\vb*{x}^{(t)}=(x^{(t)}_1,\ldots,x^{(t)}_n)$ for every $t\in\{1,\ldots,N\}$, be $N$ datasets s.t. the entries $x^{(t)}_i$ for each $i\in\{1,\ldots,n\}$ are i.i.d. samples from the discrete space $\mathcal{X}$ following the probability distribution $\pi_{\mathcal{X}}$. Let $\mathcal{C}^{(1)},\ldots,\mathcal{C}^{(N)}$ be $N$ privacy mechanisms that locally obfuscate points from $\mathcal{X}$ to $\mathcal{Y}$ such that the mechanism $\mathcal{C}^{(t)}$ is applied to the dataset $\vb*{x}^{(t)}$ and $\mathcal{C}^{(t)}_{xy}=\mathbb{P}[y|x]$ 
for all $x,y\in \mathcal{X},\mathcal{Y}$ and $i\in\{1,\ldots,N\}$. Denoting the random variable of the output when $x^{(t)}_i$ is obfuscated with $\mathcal{C}^{(t)}$ as $Y^{(t)}_i$, let 
   $(y^{(t)}_1,\ldots,y^{(t)}_n)$ be a realisation of $(Y^{(t)}_1,\ldots,Y^{(t)}_n)$ for every $t\in\{1,\ldots,N\}$. 

Let $\mathcal{G}=\begin{pmatrix}
        \mathcal{C}^{(1)}&\ldots & \mathcal{C}^{(N)}
    \end{pmatrix}$ 
be referred to as the \emph{combined mechanism} a.k.a. the \emph{output probability matrix} satisfying:
\begin{align}
    &\mathcal{G}\left(x,y^{(t)}_i\right)=\mathbb{P}\left[\left.y^{(t)}_i\right| x\right]=\mathcal{C}^{(t)}_{x,y^{(t)}_i}\nonumber\\
    &\forall x\in\mathcal{X},\,i\in\{1,\ldots,n\}\nonumber.
\end{align} 
GIBU estimates $\pi_{\mathcal{X}}$ by converging to the maximum likelihood estimate (MLE) of $\pi_{\mathcal{X}}$ with the knowledge of the noisy data and the obfuscating channels. GIBU works as follows:
\begin{enumerate}
    \item Start with any full-support PMF $\theta_0$ on $\mathcal{X}$.
    \item Iterate $ \theta_{r+1}(x)=\frac{1}{Nn}\sum\limits_{t=1}^N\sum\limits_{i=1}^n\frac{\theta_r(x)\mathcal{G}\left(x,y^{(t)}_i\right)}{\sum\limits_{z\in\mathcal{X}}\theta_r(z)\mathcal{G}\left(z,y^{(t)}_i\right)}$ for all $x\in\mathcal{X}$.
\end{enumerate}
\end{definition}
Setting $\hat{\pi}_{\mathcal{X}}=\lim\limits_{r\to\infty}\theta_{r}$ and $\vb*{y}^{t}=(y^{(t)}_1,\ldots,y^{(t)}_n)$, let $\hat{\pi}_{\mathcal{X}}$ (the MLE of the prior obtained with GIBU) be functionally denoted by: $$\texttt{GIBU}\left(\left(\mathcal{C}^{(1)},\vb*{y}^{(1)}\right),\ldots, \left(\mathcal{C}^{(N)},\vb*{y}^{(N)}\right)\right).$$

\color{black}
\begin{definition}[Earth mover's distance~\cite{kantarovich}]~\label{def:EMD}
Let $\pi_1$ and $\pi_2$ be PMFs defined over a discrete space of locations $\mathcal{X}$. For a metric $d\colon\mathcal{X}^2\mapsto \mathbb{R}_{\geq 0}$, the \emph{earth mover's distance (EMD)} (aka the \emph{Kantorovich–Rubinstein metric}) is defined as $$EMD(\pi_1,\pi_2)=\min\limits_{\mu \in \Pi(\pi_1,\pi_2)} \mu(x,y)d(x,y)$$ where $\Pi(\pi_1,\pi_2)$ is the set of all joint distributions over $\mathcal{X}^2$ such that for any $\eta\in\Pi(\pi_1,\pi_2)$, $\sum\limits_{x\in\mathcal{X}}\eta(x_0,x)=\pi_1(x_0)$ and $\sum\limits_{x\in\mathcal{X}}\eta(x,x_0)=\pi_2(x)$ for every $x_0\in\mathcal{X}$.
\end{definition}

EMD is considered a canonical way to lift a distance on a certain domain to a distance between distributions on the same domain.

\begin{definition}[Statistical utility]~\label{def:statutil}
Let $\mathcal{C}$ be a privacy mechanism  that obfuscates  data on the discrete space $\mathcal{X}$. Let $\pi_{\mathcal{X}}$ be the PMF of the original locations and let $\hat{\pi}_{\mathcal{X}}$ be its estimate by IBU. Then we define the \emph{statistical utility} of the mechanism $\mathcal{C}$ as $EMD(\hat{\pi}_{\mathcal{X}},\pi_{\mathcal{X}})$.
\end{definition}

\subsection{Optimization of MI and QoS}


\begin{definition}[Blahut-Arimoto algorithm~\cite{Blahut72computationof,Arimoto1972AnAF}]~\label{def:BA}
Let $X$ be a random variable on the discrete space $\mathcal{X}$ with PMF $\pi_{\mathcal{X}}$ and $\vb*{C}(\mathcal{X},\mathcal{Y})$ be the space of all mechanisms encoding $\mathcal{X}$ to $\mathcal{Y}$. For a distortion $d\colon\mathcal{X}\times\mathcal{Y}\mapsto\mathbb{R}_{\geq 0}$ and fixed $d^*\in\mathbb{R}^+$, we wish to find the mechanism  $\hat{\mathcal{C}}\in\vb*{C}(\mathcal{X},\mathcal{Y})$ that  minimizes MI given the bound $d^*$ on distortion:
\begin{equation*}
    \hat{\mathcal{C}}=\argmin\limits_{\substack{\mathcal{C}\in\vb*{C}(\mathcal{X},\mathcal{Y})\\AvgD(X,\mathcal{C},d)\leq d^*}}{I(X|Y_{X,\mathcal{C}})}
\end{equation*}
where, for any $\mathcal{C}\in\vb*{C}(\mathcal{X},\mathcal{Y})$, $Y_{X,\mathcal{C}}$ is the random variable on $\mathcal{Y}$ denoting the output of the encoding of $X$. The \emph{Blahut-Arimoto algorithm (BA)} provides an iterative method to construct  $\hat{\mathcal{C}}$ as follows:

\begin{enumerate}
    \item Start with any full-support PMF  $c_{0}$ on $\mathcal{X}$ and any ${\mathcal{C}^{(0)}}$.
    \item Iterate: 
    \begin{equation}\label{eq:BAIt:Channel}
        \mathcal{C}^{(t+1)}_{xy}=\frac{c_t(y)\exp{-\beta d(x,y)}}{\sum\limits_{z\in\mathcal{Y}}c_t(z) \exp{-\beta d(x,z)}}
    \end{equation}
    \begin{equation}\label{eq:BAIt:Marginal}
        c_{t+1}(y)=\sum\limits_{x\in\mathcal{X}}\pi_{\mathcal{X}}(x)\mathcal{C}^{(t+1)}_{xy}
    \end{equation}
\end{enumerate}
where $\beta>0$ is the negative of the slope of the \emph{rate-distortion} function $RD(X,d^*)=\min_{\mathcal{C}\in\vb*{C}(\mathcal{X},\mathcal{Y})}I(X|Y_{X,\mathcal{C}})$ under $AvgD(X,\mathcal{C},d)\leq d^*$. We call $\beta$  the \emph{loss parameter}, capturing the role of $d^*$ in BA. 
\end{definition}

\begin{remark}\label{rem:BATransfom}
The equations \eqref{eq:BAIt:Channel}  and \eqref{eq:BAIt:Marginal} above define two transformations ${\mathcal{F}}: \vb*{D}(\mathcal{X})  \rightarrow \vb*{C}(\mathcal{X},\mathcal{Y})$ and $\mathcal{G}: \vb*{C}(\mathcal{X},\mathcal{Y}) \rightarrow \vb*{D}(\mathcal{X})$, where $\vb*{D}(\mathcal{X})$ is the space of distributions on $\mathcal{X}$,  so that 
$\mathcal{C}^{(t+1)}= {\mathcal F}(c_t)$ and
$c_{t+1} = {\mathcal{G}}(\mathcal{C}^{(t+1)})$.
\end{remark} 

\begin{remark}\label{rem:BAConverge}
In \cite{csizar}, Csisz\'ar proved the convergence of BA when $\mathcal{X}$ is finite.  The limit $\lim_{n\rightarrow \infty} (\mathcal{F} \circ \mathcal{G})^n(\mathcal{C}^{(0)})$ is the optimal mechanism $\hat{\mathcal{C}}$ (parametrized by $\beta$), and it is uniquely determined by the prior $\pi_{\mathcal{X}}$ and by the initial PMF $c_0$.
Note that $\hat{\mathcal{C}}$ is  a fixpoint of $\mathcal{F} \circ \mathcal{G}$, i.e. $\hat{\mathcal{C}} = (\mathcal{F} \circ \mathcal{G})(\hat{\mathcal{C}})$, and that $\hat{c} = \mathcal{G}(\hat{\mathcal{C}})$ is a fixpoint of $\mathcal{G} \circ \mathcal{F}$.
\end{remark}

\begin{remark}~\label{rem:BAgeoind}
In \cite{Oya:17:CCS}, Oya et al. proved that, when $d$ is the Euclidean metric, the  mechanism $\hat{\mathcal{C}}$ obtained from BA with loss parameter $\beta$ satisfies $2\beta$-geo-ind.
\end{remark}

In the context of the location-privacy, as addressed in this work, we obfuscate the original locations to points in the same space and, hence, for the rest of the paper we consider the spaces of the secrets and the noisy locations to be the same, i.e., $\mathcal{X}=\mathcal{Y}$.

\begin{table}
\caption{Key notations}\label{table:notations}
\centering
\resizebox{1\columnwidth}{!}{%
\begin{tabular}{|c|c|}
\hline
Notation & Meaning \\ \hline
$\mathcal{X}$  & Finite space of locations\\
$\vb*{C}(\mathcal{X},\mathcal{Y})$ & Space of all mechanisms encoding $\mathcal{X}$ to $\mathcal{Y}$\\
$d^*$ & Maximum average distortion\\
$\beta$ & Loss parameter of RD function\\
$\pi_{\mathcal{X}}$ & distribution of the original locations (true prior)\\
$\hat{\pi}_{\mathcal{X}}$ & Estimation of the true prior\\
BA & Blahut-Arimoto algorithm\\
IBU & Iterative Bayesian Update\\
$\delta_{\text{BA}} $ & Precision parameter for BA to converge\\
$\delta_{\text{IBU}} $ & Precision parameter for IBU to converge\\
$N$ & Number of iterations of PRIVIC\\
$\hat{\mathcal{C}}_{\text{BA}}(\theta,N)$ & Mechanism produced by PRIVIC
\\\hline
\end{tabular}%
}
\end{table}

\section{Location-privacy with the Blahut-Arimoto algorithm}~\label{sec:LPPM_with_BA}

Definition~\ref{def:BA} shows that the BA mechanism optimizes between MI and average distortion, which is a standard choice for measuring QoS. Furthermore, Remark~\ref{rem:BAgeoind} formally links the mechanism produced by BA with geo-ind, which is our reference privacy notion.

In this section, we investigate the privacy protection offered by BA beyond  geo-ind, study the statistical utility it renders, and compare it with LAP, the canonical  mechanism for geo-ind.

\subsection{Elastic location-privacy with BA}~\label{sec:BAElastic}

One of the concerns harboured by geo-ind is that it treats the space in a uniform way, thus making isolated locations vulnerable to an attacker that knows the prior distribution. This issue has been raised and addressed by Chatzikokolakis et al. in \cite{Elastic_Kostas} where the authors introduce  a variant of LAP based on an \emph{elastic distinguishability metrics}, which they refer to as \emph{elastic mechanisms}. Such mechanisms obfuscate locations not only by considering the Euclidean distance between them but also by taking into account an abstract attribute of the reported location, called \emph{mass}, which is a parameter of the definition.  

Formally, if $\mathcal{R}_{\text{elas}}$ is an elastic mechanism with privacy parameter $\epsilon$ defined on  $\mathcal{X}$, then, for all $x,y\in \mathcal{X}$,
$\mathcal{R}_{\text{elas}}$ must satisfy:
\begin{align}
    \mathbb{P}[\mathcal{R}_{\text{elas}}(x)=y]\propto \exp{-\epsilon d_{ \text{E}}(x,y)}~\label{eq:elasticprop1}\\
    \mathbb{P}[\mathcal{R}_{\text{elas}}(x)=y]\propto q(y)~\label{eq:elasticprop2}
\end{align}
\revision{where $q$ is the probability distribution of the reported locations}.

\revision{Note that Equations \ref{eq:elasticprop1} and \ref{eq:elasticprop2} characterize the properties of an elastic mechanism  $\mathcal{R}_{\text{elas}}$, but they \emph{do not define  what $\mathcal{R}_{\text{elas}}$ exactly is, as a function.} In fact, as a definition,  Equation~\ref{eq:elasticprop2} would be circular, since it uses the probability mass $q$ generated by $\mathcal{R}_{\text{elas}}$ without knowing what $\mathcal{R}_{\text{elas}}$ is. 
As we will see, BA solves this problem
by constructing the mechanism $\mathcal{R}_{\text{elas}}$
as a fixpoint of a recursive process starting from a uniform output distribution $q$. (To be precise the process is mutually recursive, alternating the generation of a new mechanism and a new output distribution, that, in turn, is fed into BA to generate the mechanism at the next step.)}


$\mathcal{R}_{\text{elas}}$, unlike LAP, protects a point in a densely populated area (e.g. city) and a geo-spatially isolated point (e.g. island) differently by considering not only the ground distance between the true and the reported locations but also the  mass of the reported location. 
The exact mechanism depends of course on how we define the notion of mass. 
A natural way, and the  most meaningful from the privacy point of view, is to set the mass of $y$ to be the probability to be reported (from any true location $x$). 
Under this definition, the interpretation  of \eqref{eq:elasticprop2} is in the spirit of
obtaining privacy by ensuring that the set of possible true locations (given the reported one) is large. In other words,  
given a  true location $x$, we tend to report with higher probability those locations $y$ that are reported with high probability from other locations as well so that it becomes harder to re-identify $x$ as the original one. Note that this property   is not incompatible with  the geo-ind guarantee. However, LAP does not provide it.   

Obviously, the definition of mass as the probability to be reported would be circular, because it would depend on the mechanism, which in turn is defined in terms of the mass. 
The authors of  \cite{Elastic_Kostas} do not explain how this mechanism could be constructed. Fortunately, the following theorem shows that an elastic mechanism of this kind can be constructed using BA.  The proof is provided in Appendix~\ref{app:proofs}.

\begin{restatable}{theorem}{BAElastic}~\label{th:BAElastic}
The privacy mechanism generated by BA produces an elastic location-privacy mechanism.
\end{restatable}

Note also that there can be many mechanisms satisfying \eqref{eq:elasticprop1} and \eqref{eq:elasticprop2} (also with the mass interpreted as probability). The one produced by the BA is the mechanism that offers the best QoS among these. Finally, a consequence of the connection with BA is that it provides an understanding of the elastic mechanism in terms of information theory and of the attacker illustrated in the previous section.

\subsubsection*{Experimental validation}

Having furnished the theoretical foundation, we now enable ourselves to empirically validate that BA, indeed, satisfies the properties of the elastic mechanism unlike LAP, its state-of-the-art geo-indistinguishable counterpart. We perform experiments using real location data from the Gowalla dataset~\cite{Gowalla:online, cho2011friendship}. We consider 10,078 Gowalla check-ins from a central part of Paris bounded by latitudes (48.8286, 48.8798) and longitudes (2.2855, 2.3909) covering an area of 8Km$\times$6Km discretized with a $16\times12$ grid. 

\begin{wrapfigure}{L}{0.6\columnwidth}
\vspace{-12pt}
\centering
\includegraphics[width=0.7\columnwidth]{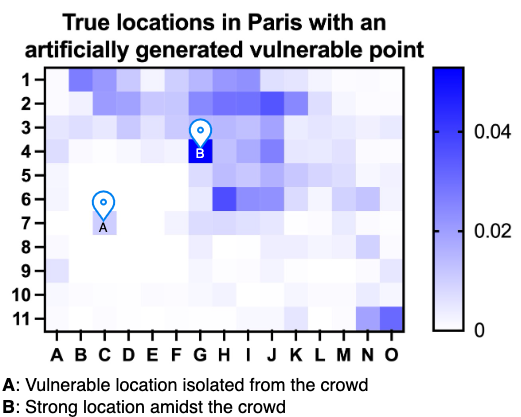}
\caption{Gowalla check-in locations in Paris with an artificially planted vulnerable point, $A$, in isolation, and a strong point, $B$, in a crowded area.}
\label{fig:str_vul}
\vspace{-8pt}
\end{wrapfigure}

In order to demonstrate the property of an elastic mechanism, we artificially introduced an ``island'' amidst the locations in Paris by choosing a grid $A$ in a low-density area of the dataset (in the south-west region), assigning the probability mass of the grids around $A$ to 0, and dumping this cumulative mass from the surrounding region to $A$, ensuring that the sum of the probability masses of all the grids remains to be 1. We call $A$ as a \emph{vulnerable location} in the map as it is isolated from the crowded area. To visualize the elastic behaviour of the mechanisms for locations in crowded regions, we consider another grid $B$ in the central part of the map which has a high probability mass and has a highly populated surrounding -- we refer to such a grid $B$ as a \emph{strong location} in the map. Figure~\ref{fig:str_vul} illustrates the selection of vulnerable and strong locations in the Paris dataset.

\begin{figure*}[htbp]
\centering
\begin{subfigure}[b]{1\columnwidth}
   \includegraphics[width=1\columnwidth]{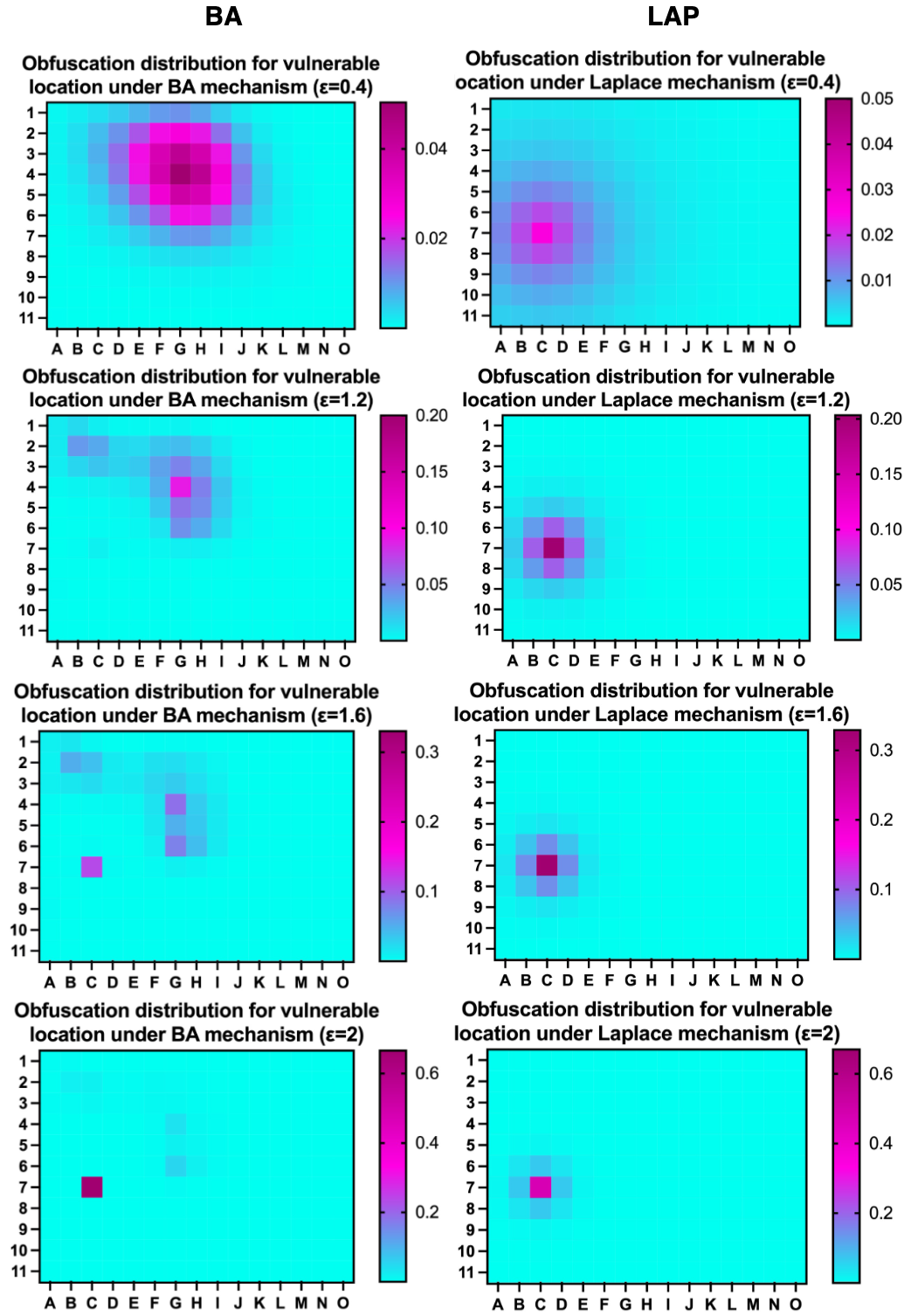}
   \caption{Reporting distribution of the vulnerable location $A$}
   \label{fig:Elastic_vul} 
\end{subfigure}
\begin{subfigure}[b]{1\columnwidth}
   \includegraphics[width=1\columnwidth]{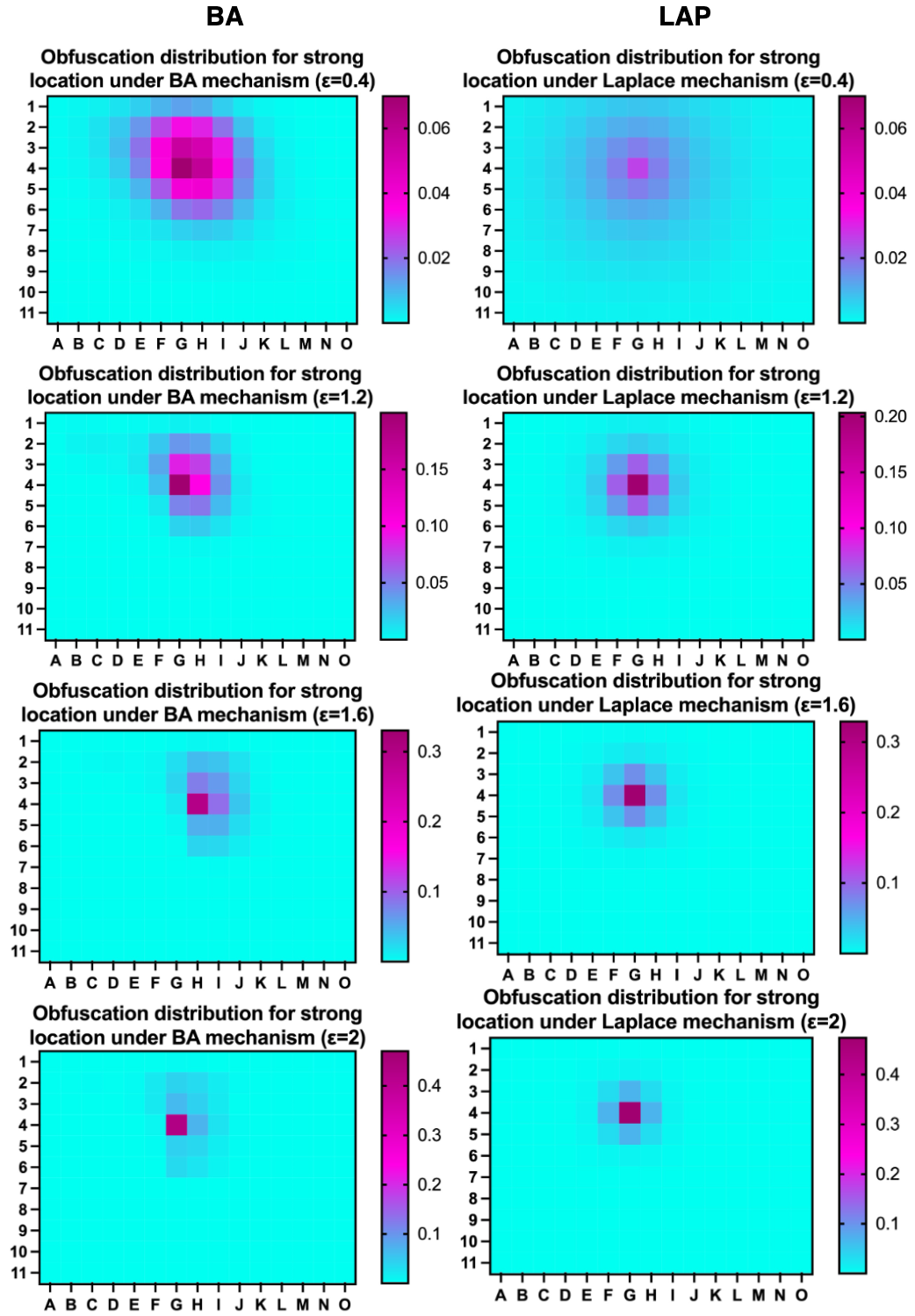}
   \caption{Reporting distribution of the strong location $B$}
   \label{fig:Elastic_str}
\end{subfigure}
\caption{Distribution of privatizing the vulnerable and the strong locations for different levels of privacy. Top-down, the rows illustrate the results for $\epsilon=0.4,1.2,1.6,2$, respectively.}
\label{fig:ElasticBAvLap}
\end{figure*}

For the mechanism derived from BA with a loss parameter $\beta$, we know, by Remark~\ref{rem:BAgeoind}, that the privacy parameter $\epsilon$ is $2\beta$, which we use to tune the privacy level of LAP in order to 
compare the two mechanisms under the same level of geo-ind. Figure~\ref{fig:ElasticBAvLap} illustrates the probability distribution of reporting a privatized point on the map by obfuscating the vulnerable and the strong locations with different levels of geo-ind -- we vary the value of $\epsilon$ to be $0.4,1.2,1.6,2$.

By comparing with the distribution of the true locations in Paris given by Figure~\ref{fig:str_vul}, we observe that when the value of $\epsilon$ is low (privacy is high), the reported location with BA is likely to be mapped to a nearby densely populated place. For example, with $\epsilon=0.2$, the highest level of privacy considered in the experiments, the location reported by BA will most probably be around the most crowded region of Paris. As $\epsilon$ increases, the location most likely to be reported by BA systematically moves to a densely populated region closer and closer to the true vulnerable location. LAP, on the other hand, always obfuscates every location around its true position in the map -- varying the value of $\epsilon$ changes the spread of the distribution around the true location. As explained in the introduction, this might be problematic as the vulnerable location is known to be isolated and, hence, even being reported somewhere nearby would potentially result in its re-identification.  

For example, we would like to highlight the setting of $\epsilon=1.6$ for the vulnerable location to show that the distribution of the location reported by LAP is almost completely around the true vulnerable point covering an area that is deserted, i.e., there is no realistic chance of someone being located in that region. Thus, despite providing formal $1.6$-geo-ind, LAP fails to protect such a vulnerable location from being potentially identified. BA, on the other hand, does the job quite well, adhering to the principles of the elastic mechanism -- it distributes the reported location in the crowded areas nearby providing a sense of camouflage amidst the many possibilities, in addition to $1.6$-geo-ind. 

In the case of privatizing the strong location, Figure~\ref{fig:Elastic_str} shows that both BA and LAP behave similarly by concealing the point around its true position. This would not give rise to a similar issue as for the vulnerable location because, by definition, the strong location $B$ is already positioned in a highly dense region of the map and, hence, being privatized, it will still remain among the crowd with a high probability. 

Focusing on the utility of individual users, we note that due to theories from Nash equilibrium~\cite{nash} and Hotelling's spatial competition~\cite{GALOR_Hotelling}, a huge fraction of the typical points of interest (POIs) like cinemas, theatres, restaurants, retails, etc. lie in crowded areas syncing with the distribution of population. Therefore, for an isolated point in the map that is located in some extremely unpopulated area (e.g. some forest or island far from the city), the closest POI is usually going to be in the nearest urban region, i.e., a region on the map with a high density of population. Suppose $A$ is one such isolated location and let $A_{\text{BA}}$ and $A_{\text{LAP}}$ be the reported locations for $A$ obfuscated with BA and LAP, respectively. Due to the elastic property of BA, $A_{\text{BA}}$ is likely to be at a nearby crowded location to $A$, while $A_{\text{LAP}}$ is likely to be around the true location $A$. Let $P_{\text{BA}}$ and $P_{\text{LAP}}$ be the nearest POIs from the reported locations $A_{\text{BA}}$ and $A_{\text{LAP}}$, respectively. The most likely scenario is that $P_{\text{BA}}$ and $P_{\text{LAP}}$ are almost at a similar place under the assumption that typical POIs follow the distribution of the crowd and, therefore, a vulnerable user has to travel a similar distance from their true position in both the cases, except that under LAP, the privacy of $A$ will be compromised much more than that   under BA.

\subsection{Statistical utility: BA vs LAP}~\label{sec:BAvsLapStatUtil}

\begin{figure}[htbp]
\centering
\begin{subfigure}[b]{1\columnwidth}
   \includegraphics[width=0.9\columnwidth]{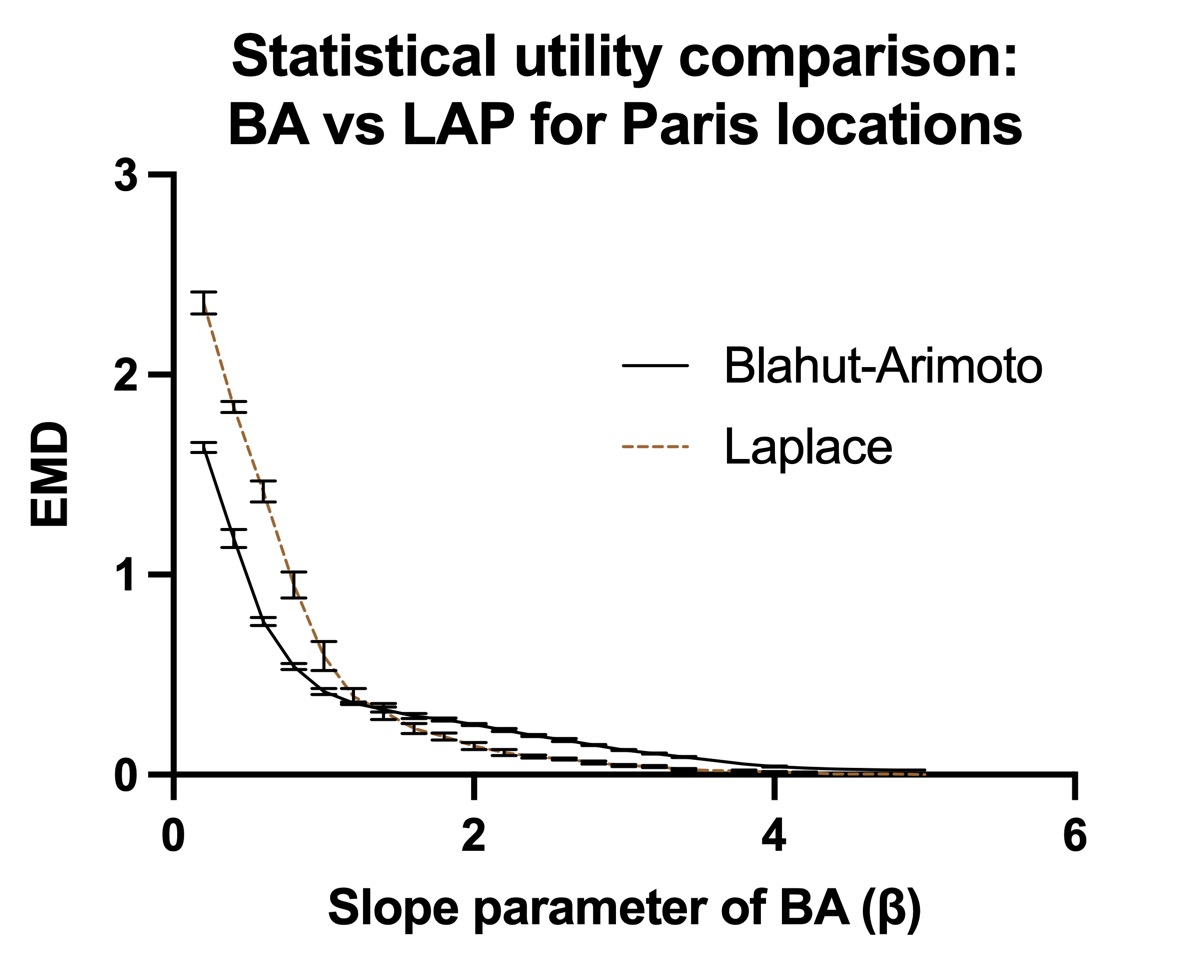}
   \caption{Statistical utility for BA and Laplace on locations in Paris}
   \label{fig:StatUtilBAvsLapParis} 
\end{subfigure}
\begin{subfigure}[b]{1\columnwidth}
\includegraphics[width=0.9\columnwidth]{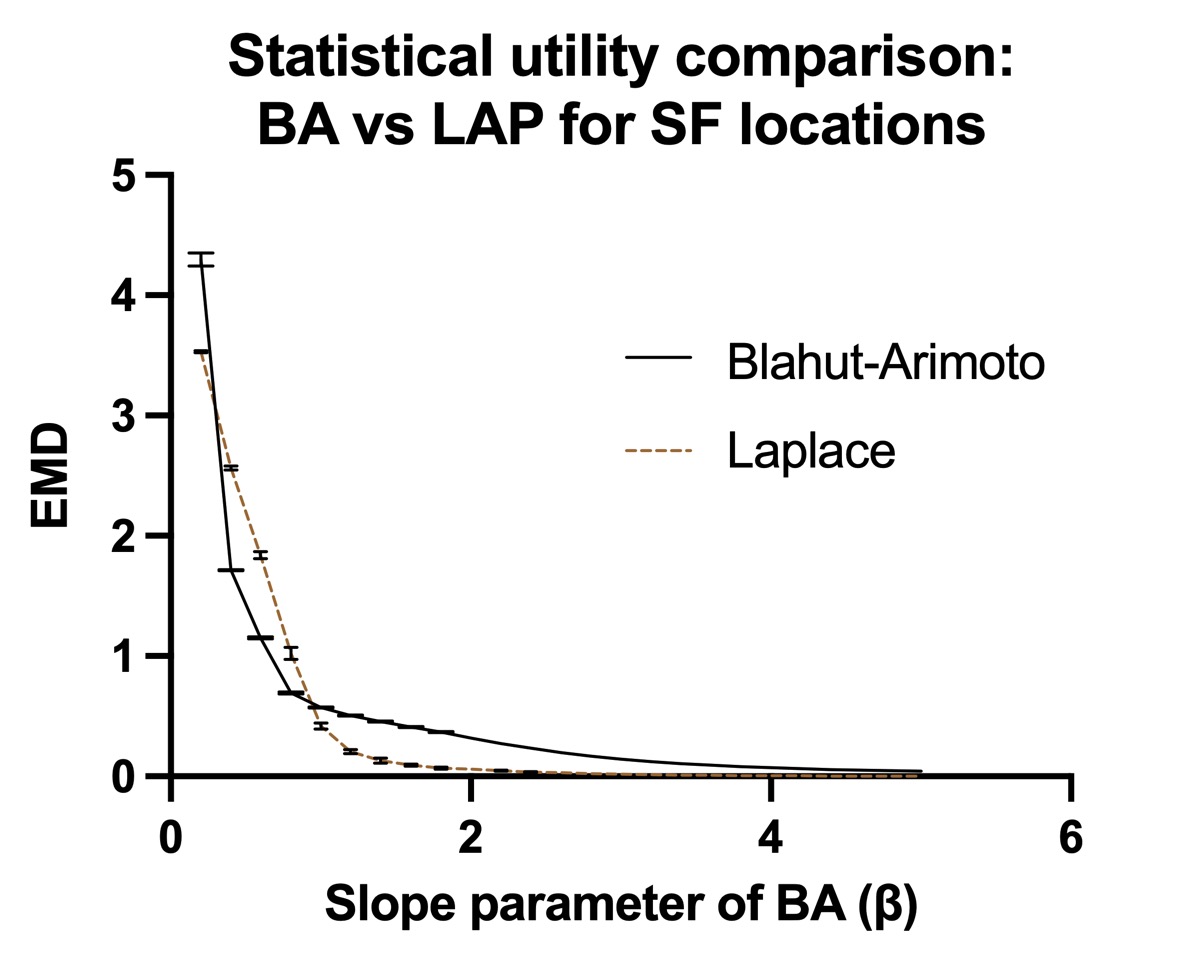}
   \caption{Statistical utility for BA and Laplace on locations in SF}
   \label{fig:StatUtilBAvsLapSF}
\end{subfigure}
\caption{Statistical utility in terms of \emph{earth mover's distance} (EMD) between the true and the estimated distributions for locations in Paris and San Francisco, under BA and LAP.}
\label{fig:StatUtilBAvsLap}
\end{figure}
Now we proceed to empirically compare the statistical utility of BA and LAP by performing experiments on the locations obtained from the Gowalla dataset for two different cities: Paris and San Francisco. In addition to the same setting for the Gowalla check-ins in Paris as considered in the experiments of Section~\ref{sec:BAElastic}, here we also test for 123,025 check-in locations from the Gowalla dataset in a northern part of San Francisco bounded by latitudes (37.7228, 37.7946) and longitudes (-122.5153, -122.3789) covering an area of 12Km$\times$8Km discretized with a 24$\times$17 grid. The locations were privatized with BA and LAP under varying levels of privacy -- the loss parameter, $\beta$, for BA ranged from $0.2$ to $5.0$, which implies that the value of the geo-ind parameter, $\epsilon$, ranged from $0.4$ (very high level of privacy) to $10.0$ (almost no privacy). To account for the randomness in the process of generating the sanitized locations, $5$ simulations were run for each value of the privacy parameter for obfuscating every location in both datasets.

Figure~\ref{fig:StatUtilBAvsLap} reveals that BA possesses a significantly better statistical utility than LAP for a high level of privacy (for $\beta\in (0.4,1.4]$ and $\beta\in(0,1)$, i.e., $\epsilon$ up to $2.8$ and $2$, in Paris and San Francisco datasets, respectively). As the level of privacy decreases, the EMD 
of BA becomes worse than that of LAP. We conjecture that this is the price to pay for the added privacy provided by the elasticity of the mechanism. Eventually, the EMD  between the true and the estimated PMFs converge to $0$ in both mechanisms, as we would expect, fostering the maximum possible statistical utility with, practically, no privacy guarantee.

Summarizing the results from Sections~\ref{sec:BAElastic} and~\ref{sec:BAvsLapStatUtil}, we can establish that:
\begin{itemize}
    \item in addition to providing a formal geo-ind guarantee, BA also gives an LPPM with an elastic distinguishability metric to enhance the privacy of vulnerable locations.
    \item BA optimizes the trade-off between  QoS and MI.
    \item the statistical utility for high levels of privacy is significantly better for BA than LAP.
\end{itemize}

Therefore, we conclude that BA is a key contender for providing a comprehensive notion of location privacy while preserving the utility of the data for both the users and the service providers.

\section{Duality between IBU and BA}~\label{sec:duality}

We now explore a relationship between BA and IBU which we found rather intriguing. 
For a metric space $(\mathcal{X},d)$, let $X$ be a random variable on $\mathcal{X}$ with PMF $\pi_{\mathcal{X}}$.
Recalling the iteration of BA from \eqref{eq:BAIt:Marginal} and \eqref{eq:BAIt:Channel}:
$$c_t(y)=\sum_{x\in\mathcal{X}}\pi_{\mathcal{X}}(x)\mathcal{C}^{(t)}_{xy} \text{ and }C^{t+1}_{xy}=\frac{c_t(y)\exp{-\beta d(x,y)}}{\sum_{z\in\mathcal{X}}c_t(z) \exp{-\beta d(x,z)}}$$

Hence, we obtain:
\begin{align}
c_{t+1}(y)=\sum_{x\in\mathcal{X}}\pi_{\mathcal{X}}(x)\mathcal{C}^{(t+1)}_{xy}=\sum_{x\in\mathcal{X}}\pi_{\mathcal{X}}(x)\frac{c_t(y)\exp{-\beta d(x,y)}}{\sum_{z\in\mathcal{X}}c_t(z) \exp{-\beta d(x,z)}}\label{eq:dual1}
\end{align}

\begin{wrapfigure}{L}{0.5\columnwidth}
\vspace{-12pt}
\centering
\includegraphics[width=0.6\columnwidth]{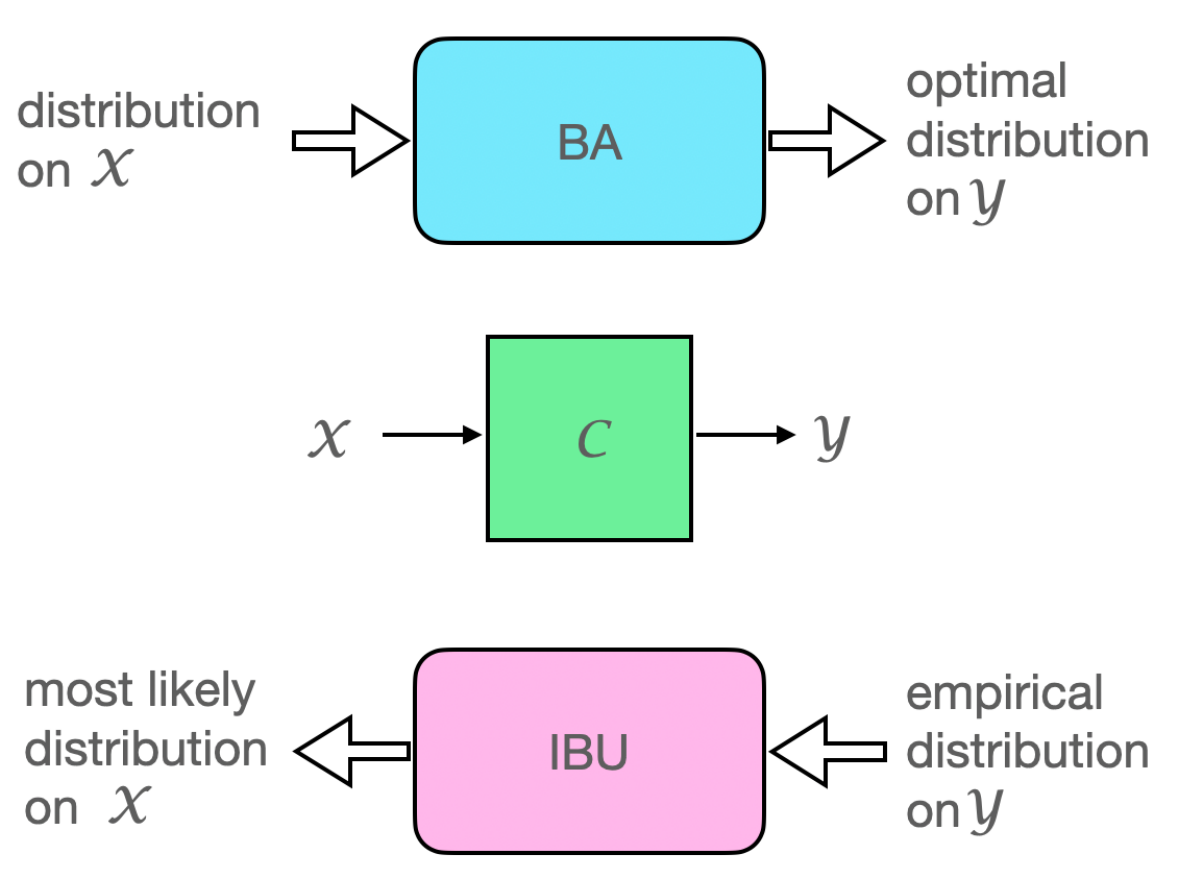}
\caption{Illustration of the duality between BA and IBU.}
\label{fig:BA-IBU}
\vspace{-10pt}
\end{wrapfigure}
Comparing it with the iteration of IBU as in Definition \ref{def:IBU}, we observe that \eqref{eq:dual1} BA is dual to IBU. Indeed,  consider an exponential mechanism of the form  $\mathcal{C} =c\exp{-\beta d(x,y)}$.
Flipping the roles of $x$ and $y$ in \eqref{eq:dual1}, and replacing the input distribution  $\pi_{\mathcal{X}}$ with the empirical distribution in output to  $\mathcal{C}$, we obtain the iterative step of IBU. 

Due to this duality between BA and IBU (illustrated in Figure~\ref{fig:BA-IBU}) and taking advantage of the fact that BA converges~\cite{csizar}, i.e., $\lim_{t\to\infty}c_t$ exists, we obtain 
that also IBU converges.

\section{PRIVIC: a privacy-preserving method for incremental data collection}~\label{sec:PRIVIC}

To ensure that the produced mechanism is truly optimal,  BA needs a good approximation of the prior distribution. 
In the beginning, we cannot assume to have such knowledge, but as the service providers incrementally collect data from their users, we can use these data to refine the estimation of the prior and get a better mechanism. These data, however, are obfuscated by the privacy mechanism and, hence, it is not obvious that the estimation of the prior really improves in the process.  We show that this is the case, and, summarizing all results obtained for BA so far, we propose a method that facilitates the service providers to incrementally collect data and gradually achieve a high statistical utility with respect to the QoS. We shall refer to our proposed method for \textbf{PRIV}acy-preserving \textbf{I}ncremental \textbf{C}ollection of location data as \emph{PRIVIC}.  

\revision{The goal of PRIVIC is to construct an obfuscation mechanism that guarantees formal geo-ind, acts as an elastic mechanism, and eventually optimizes between MI and QoS, while producing, at the same time, a good estimation of the distribution on the data.}

We shall consider locations sampled from a finite space $\mathcal{X}=\{x_1,\ldots,x_m\}$. Let the \emph{true distribution} or \emph{true PMF} on $\mathcal{X}$ (from which the users' locations are sampled) be $\pi_{\mathcal{X}}$. Note that \emph{we do not assume the knowledge of $\pi_{\mathcal{X}}$ in our method}. 
We assume that the new locations are sampled independently from the previous ones. This hypothesis is reasonable if the collection of the new data is enough separated in time from the previous one, otherwise, we would have a potential correlation between samplings due to the possibility that a user sends repeated check-ins from spatially closed locations. In any case,  
geo-ind, like DP, satisfies the property of sequential compositionality~\cite{galli2022group}, which means that privacy degradation is under control. 

In this work, to achieve geo-ind, we shall adhere to the Euclidean metric $d_{\text{E}}$ to measure the ground distance between locations.  

PRIVIC proceeds as follows (cf. also Figure~\ref{fig:privic}): 

\begin{enumerate}
    \item[1.] \revision{Set  $\theta_0, c_0$ to be the uniform distributions on $\mathcal{X}$, i.e., $\theta_{0}(x) = c_0(x)=1/|\mathcal{X}|$ for all $x\in\mathcal{X}$.}. 
    \item[2.] In step $t\geq 1$:
    \begin{enumerate}
        \item[i)] For a fixed the maximum average distortion, set $\hat{\mathcal{C}}^{(t)}=\texttt{BA}\left(\theta_{t-1},c_0\right)$.
        \item[ii)] Sample a new set of locations $\vb*{x}^{(t)}$ from the (unknown) true distribution and obfuscate
        them locally by the mechanism $\hat{\mathcal{C}}^{(t)}$ to get $\vb*{y}^{(t)}$, thus obtaining  the empirical distribution of the reported locations $\vb*{q}_t=\{\vb*{q}_t(x)\colon x\in\mathcal{X}\}$.
         \item[iii)] $\mu_t=\texttt{IBU}\left(\hat{\mathcal{C}}^{(t)},\theta_{t-1},\vb*{q}_t\right)$.
         \item[iv)] if $t=1$ then $\theta_t=\mu_t$ else $\theta_t=\mu_t\oplus \theta_{t-1}$ (combination of previous and new estimation proportional to the respective number of samples).
    \end{enumerate}
    \vspace{-0.35cm}
    \revision{\item[3.] $\hat{\pi}_{\mathcal{X}}=\texttt{GIBU}\left(\left(\hat{\mathcal{C}}^{(1)},\vb*{y}^{(1)}\right),\ldots, \left(\hat{\mathcal{C}}^{(N)},\vb*{y}^{(N)}\right)\right).$}
\end{enumerate}

\begin{figure}[htbp]
\centering
\includegraphics[width=0.9\columnwidth]{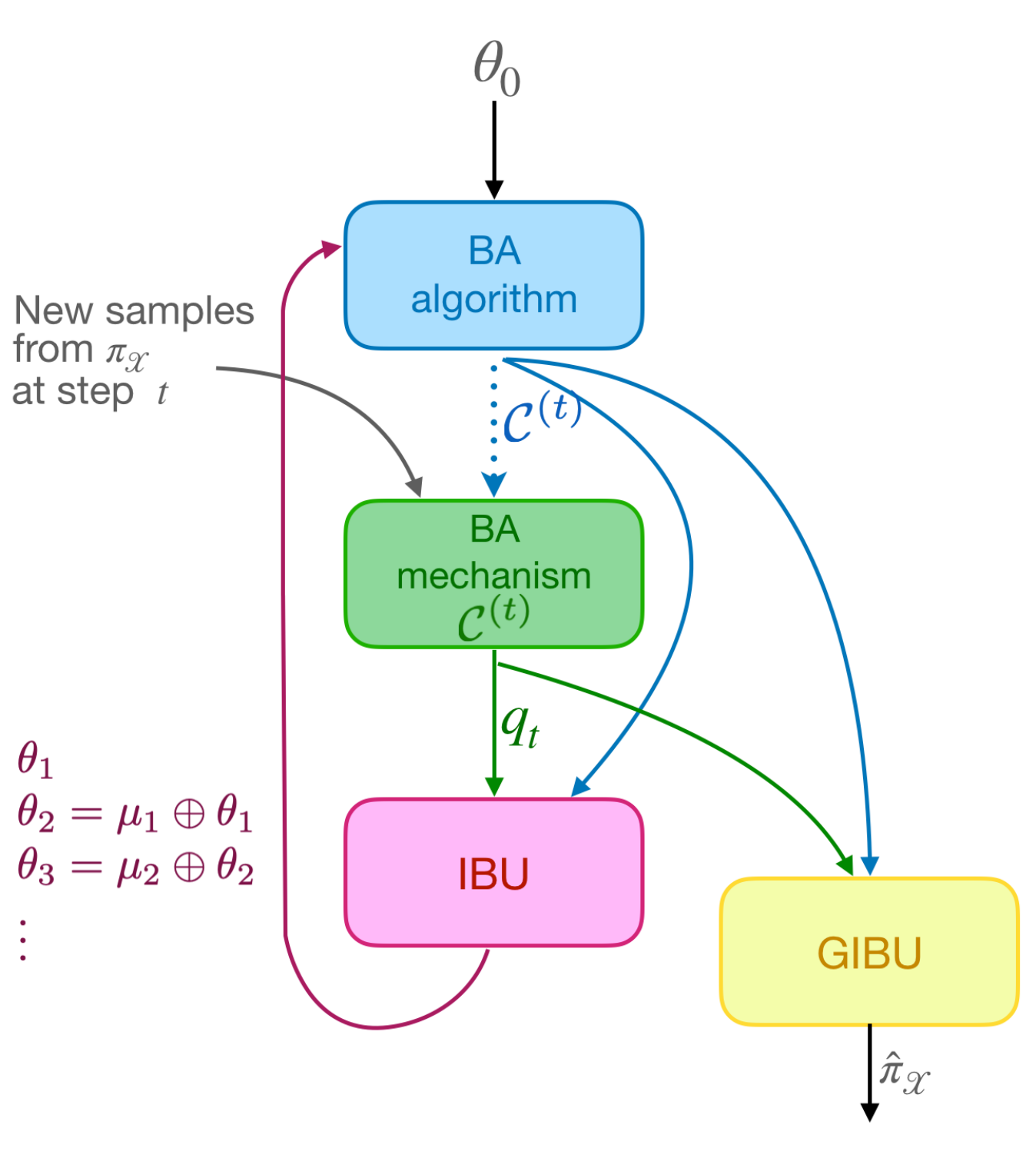}
\caption{Illustration of the iterative process of PRIVIC}
\label{fig:privic}
\end{figure}

\begin{algorithm}[ht]
\SetAlFnt{\small}
 \textbf{Input:} 
Loss parameter: $\beta$, No. of iterations: $N$, precision of BA: $\delta_{\text{BA}}$, precision of IBU: $\delta_{\text{IBU}}$, precision of GIBU: $\delta_{\text{GIBU}}$\;
\textbf{Output:} 
\revision{Optimal channel: $\hat{C}$,}
Estimation of true PMF:$\hat{\pi}_{\mathcal{X}}$\; 
\revision{$\theta_0(x) \leftarrow \nicefrac{1}{|{\mathcal{X}}|}$\;
$c_0 \leftarrow \nicefrac{1}{|{\mathcal{X}}|}$\;}
$t\leftarrow 0$\;
\While {$t \leq N$}
    {
    $\hat{\mathcal{C}}^{(t+1)}=\Call{\texttt{BA}}{\theta_{t},c_0,\beta,\delta_{\text{BA}}}$\;
    $\vb*{y}^{(t)}\leftarrow\left(y^{(t)}_1,\ldots,y^{(t)}_n\right)$: New noisy locations reported by users after obfuscating their newly sampled  true locations with $\hat{\mathcal{C}}^{(t)}$\;
    $\vb*{q}\leftarrow\{q(x)\colon x\in\mathcal{X}\}$: Empirical PMF obtained from $\mathcal{L}$ by the service provider\;
    $\mu\leftarrow \Call{\texttt{IBU}} {\hat{\mathcal{C}}^{(t+1)},\theta_{t},\vb*{q},\delta_{\text{IBU}}}$\;
     if $t=0$ then
    $\theta_{t+1}\leftarrow \mu$ else $\theta_{t+1}\leftarrow \mu \oplus \theta_t$\;
    $t\leftarrow t+1$\;
    }
 \revision{
$\hat{\pi}_{\mathcal{X}}\leftarrow\texttt{GIBU}\left(\left(\hat{\mathcal{C}}^{(1)},\vb*{y}^{(1)}\right),\ldots, \left(\hat{\mathcal{C}}^{(N)},\vb*{y}^{(N)}\right),\delta_{\text{GIBU}}\right)$\;
$\hat{\mathcal{C}}\leftarrow \text{BA}(\hat{\pi}_{\mathcal{X}}, c_0, \beta, \delta_{\text{BA}})$\;
}
\textbf{Return:} \revision{$\hat{\mathcal{C}}$,} $\hat{\pi}_{\mathcal{X}}$
\caption{PRIVIC}\label{alg:PRIVIC}
\end{algorithm}

\begin{algorithm}
\SetAlFnt{\small}
\textbf{Input:} PMF: $\pi$, initial mechanism: $\mathcal{C}^{(0)}$, loss parameter: $\beta$, precision: $\delta_{\text{BA}}$\;
\textbf{Output:} mechanism giving minimum mutual information for maximum avg. distortion encapsulated by $\beta$: $\hat{\mathcal{C}}$\;
\SetKwFunction{Fmain}{BA}
\SetKwProg{Fn}{Function}{:}{}
\Fn{\Fmain{$\pi,c_0,\beta,\delta_{\text{BA}}$}}{ 
$t\leftarrow 0$\;
\While{$\delta_{\text{BA}} \leq |C^{(t)} - C^{(t-1)}|$}
    {
    $\mathcal{C}^{(t+1)}_{xy}\leftarrow \frac{c_t(y)\exp{-\beta d_{\text{E}}(x,y)}}{\sum\limits_{z\in\mathcal{X}}c_t(z)\exp{-\beta d_{\text{E}}(z,y)}}$\;
    $c_{t+1}(y)\leftarrow\sum\limits_{x\in\mathcal{X}}\pi(x)\mathcal{C}^{(t+1)}_{xy}$\;
    $t\leftarrow t+1$
    }
$\hat{\mathcal{C}}\leftarrow \mathcal{C}^{(t)}$\;
\textbf{Return:} $\hat{\mathcal{C}}$
}
\caption{Blahut-Arimoto algorithm (BA)}\label{alg:BA}
\end{algorithm}

\begin{algorithm}
\SetAlFnt{\small}
\textbf{Input:} Privacy mechanism: $\mathcal{C}$, Full-support PMF: $\vartheta_0$, empirical PMF from observed data: $\vb*{q}$, precision: $\delta_{\text{IBU}}$\;
\textbf{Output:} MLE of true PMF: $\theta$\;
\SetKwFunction{Fmain}{IBU}
\SetKwProg{Fn}{Function}{:}{}
\Fn{\Fmain{$\mathcal{C},\vartheta_0,\vb*{q},\delta_{\text{IBU}}$}}{ 
Set $t\leftarrow 0$\;
\While{$\delta_{\text{IBU}}< |\vartheta_t-\vartheta_{t-1}|$}
{
    $\vartheta_{t+1}(x)\leftarrow\sum\limits_{y\in\mathcal{X}}\vb*{q}(y)\frac{\mathcal{C}_{xy}\vartheta_t(x)}{\sum\limits_{z\in \mathcal{X}}\mathcal{C}_{zy}\vartheta_t(z)}$\;
    $t\leftarrow t+1$
}
${\theta}\leftarrow\vartheta_{t}$\;
\textbf{Return:} ${\theta}$\;
}
\caption{iterative Bayesian update (IBU)}\label{alg:penaltyAlg}
\end{algorithm}

\revision{
\begin{remark}
The initial distribution $\theta_{0}$ does not need to be a uniform distribution, any fully-supported distribution would suffice for the process to eventually converge to an optimal mechanism. However, starting with a uniform distribution allows us 
to avoid any bias in the mechanisms produced in the intermediate steps.  
\end{remark}
}
\revision{
\begin{remark}We believe that the last step (3) is not really necessary: The combination of all estimations should already be the MLE of the true distribution, and this is also what we have witnessed in the experiments. However, applying this last step allows us  to \emph{formally prove} the converge to the MLE, using the results for GIBU in \cite{EhabGIBU}.
\end{remark}
}

\revision{In the practical implementation of PRIVIC, we use the precision parameters $\delta_{\text{BA}}$, $\delta_{\text{IBU}}$, and $\delta_{\text{GIBU}}$ to set the threshold of empirical convergence of BA, IBU, and GIBU, respectively.} Let the privacy mechanism generated this way after $N$ iterations, for fixed parameters $c_0$, $\beta$, $\delta_{\text{BA}}$, $\delta_{\text{IBU}}$, and $\delta_{\text{GIBU}}$, be functionally represented as $\hat{\mathcal{C}}_{\text{BA}}\left (\theta_0,N\right )$.

Concerning statistical utility, it is important to ensure that IBU converges to the true distribution. As a matter of fact, IBU always converges to an MLE but the MLE may not be unique~\cite{EhabGIBU}. More precisely, there can be more than one distribution that is the most likely input to the obfuscation mechanism, for a given empirical distribution on the noisy data. Thus, even though IBU  converges, it may converge to a distribution different from the true one. This is a problem in the method by Oya et al. in \cite{Oya:19:EuroSnP} which 
computes the obfuscation mechanism via the algorithm of 
Shokri et al.~\cite{ShokriQuantifyingLocPriv2011}. The resulting mechanism optimizes the trade-off between distortion and a Bayesian notion of privacy, but may not have a unique MLE, as illustrated in the example below. 
They probably did not realize the problem, because they relied on the flawed results by \cite{AgarwalIBU} according to which every mechanism would have a unique MLE. 

\revision{The following example is a simplified version of the example given in \cite{EhabGIBU} (Sections 3.1 and 3.2.) which was aimed at showing the non-uniqueness of the MLE, and consequent convergence to the wrong distribution, in a more general setting. However, for the scope of our work, a simpler variant suffices.}

\begin{example}\label{examp:BadLPPM}
Consider three collinear locations, $a$, $b$ and $c$, where $b$ lies in between $a$ and $c$ at a unit distance from each of them. Assume that the prior distribution on these three locations is uniform and that the constraint on the utility is that it should not exceed $\nicefrac{2}{3}$. Then a mechanism that optimizes the QoS in the sense of \cite{ShokriQuantifyingLocPriv2011} is the one that maps all locations to $b$. However, this mechanism has no statistical utility, as the $b$'s  do not provide any information about the original distribution. Indeed, given $n$ obfuscated locations (i.e., $n$ $b$'s) all distributions on $a$, $b$ and $c$ of the form $\nicefrac{k_a}{n}, \nicefrac{k_b}{n}, \nicefrac{k_c}{n}$ with $k_a+k_b+k_c=n$, have the same likelihood to be the original one. 
\end{example}

Fortunately, our method does not have this problem, because the BA produces an invertible mechanism, and invertibility implies the uniqueness of the MLE \cite{EhabGIBU}. \revision{In particular, we are now able to show the convergence of PRIVIC as a  whole using the results of \cite{EhabGIBU}.}

\begin{restatable}{theorem}{BAisInvertible}\label{th:BA_is_invertible}
For any $t\geq 1$, the mechanism generated by BA over $\mathcal{X}$ at the $t$'th iteration, seen as a stochastic matrix,  is invertible. 
\end{restatable}

\begin{proof}
    In Appendix~\ref{app:proofs}.
\end{proof}


\color{black}
\begin{restatable}{theorem}{PRIVICConverges}\label{th:PRIVIC_converges}
PRIVIC converges to the unique MLE of the true distribution. 
\end{restatable}

\begin{proof}
    In Appendix~\ref{app:proofs}.
\end{proof}

\color{black}

To evaluate the statistical utility of $\hat{\mathcal{C}}_{\text{BA}}\left (\theta_0, N\right)$ (cf. Section~\ref{sec:experiments}), we will measure the EMD between the true and the estimated PMFs at the end of $N$ iterations of PRIVIC. Thus, the quantity $EMD(\hat{\pi}_{\mathcal{X}},\pi_{\mathcal{X}})$ parameterizes the utility of $\hat{\mathcal{C}}_{\text{BA}}\left (\theta_0,N\right)$ for the service providers. We use the same Euclidean distance as the underlying metric for computing, both, the EMD and the average distortion -- this consistency threads together and complements the notion of \emph{utility} of the service providers and that from the sense of the QoS of the users. 

\section{Experimental analysis of PRIVIC}~\label{sec:experiments}

In this section, we describe the empirical results obtained by carrying out experiments to illustrate and validate the working of our proposed method. Standard Python packages were used to run the experiments in a MacOS Ventura 13.2.1 environment with an Intel core i9 processor and 32 GB of RAM. Like in the previous experiments to compare the statistical utilities of BA and LAP, as elaborated in Section~\ref{sec:BAvsLapStatUtil}, we use real locations from the same regions in Paris and San Francisco from the Gowalla dataset~\cite{Gowalla:online, cho2011friendship}. In particular, we consider Gowalla check-ins from (i) a northern part of San Francisco bounded by latitudes (37.7228, 37.7946) and longitudes (-122.5153, -122.3789) covering an area of 12Km$\times$8Km discretized with a 24$\times$17 grid; (ii) a central part of Paris bounded by latitudes (48.8286, 48.8798) and longitudes (2.2855, 2.3909) covering an area of 8Km$\times$6Km discretized with a 16$\times$12 grid. In this setting, we work with 123,108 check-in locations in San Francisco and 10,260 check-in locations in Paris. Figure~\ref{fig:checkins} shows the particular points of check-in from Paris and San Francisco and Figure~\ref{fig:density} highlights their distribution. 


\begin{figure}[htbp]
\centering
   \begin{subfigure}[b]{1\columnwidth}
   \includegraphics[width=1\columnwidth]{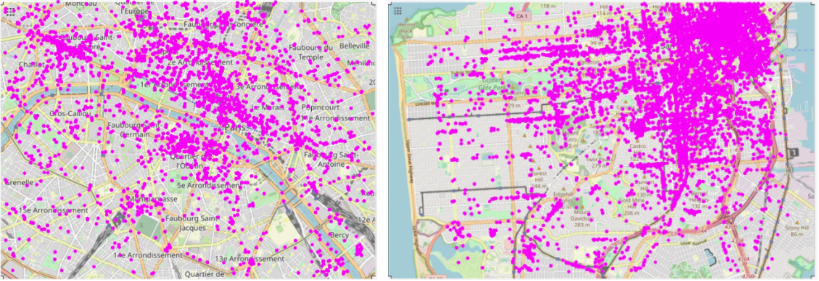}
   \caption{Check-in locations in Paris and San Francisco}
   \label{fig:checkins} 
\end{subfigure}
\begin{subfigure}[b]{1\columnwidth}
   \includegraphics[width=1\columnwidth]{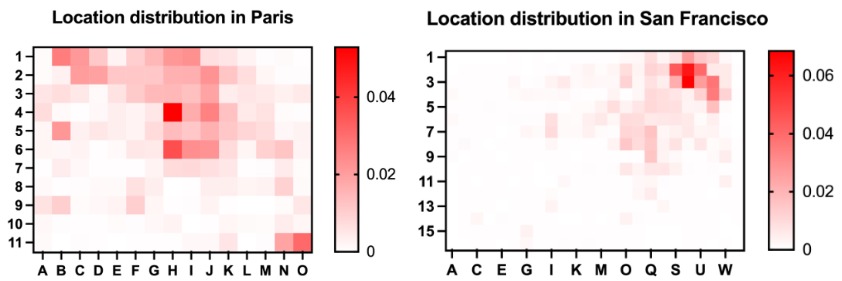}
   \caption{Density of the original locations from Paris and San Francisco}
   \label{fig:density}
\end{subfigure}
\caption{(a) visualizes the original locations from Gowalla dataset from Paris and San Francisco. (b) illustrates a heatmap representation of the locations in the two cities to capture the distribution of the data.}
\label{fig:OrLoc:SF_Paris}
\end{figure}

\begin{table}[htbp]
\centering
\caption{Run-time and complexity of BA and IBU in each cycle of PRIVIC}~\label{table:runtime}
\resizebox{1\columnwidth}{!}{%
  \begin{tabular}{|c|cc|cc|}
    \hline
    \multirow{2}{*}{Dataset} & 
    \multicolumn{2}{c|}{BA} & 
    \multicolumn{2}{c|}{IBU} \\
    & Mean run-time (sec.) & Complexity & Mean run-time (sec.) & Complexity \\
    \hline
    Paris & 3.256 & $\mathcal{O}(n^2)$ & 1.30 & $\mathcal{O}(n^2)$ \\
    \hline
    San Francisco & 16.805 & $\mathcal{O}(n^2)$ & 128.192 & $\mathcal{O}(n^2)$ \\
    \hline
    \multicolumn{5}{|c|}{\emph{Framework: MacOS Ventura 13.2.1 with Intel core i9 processor and 32 GB RAM}}\\
    \hline
  \end{tabular}}
\end{table}

We implemented PRIVIC on the locations from Paris and San Francisco separately to judge its performance on real data with very different priors. In both cases, we ran our mechanism until it empirically converged. 15 cycles of PRIVIC were required for the Paris dataset where each cycle comprised 8 iterations of BA until it converged to generate the privacy mechanism and 10 iterations of IBU until it converged to the MLE of the prior. For the San Francisco dataset, PRIVIC needed 8 cycles to converge with 5 iterations of BA and IBU each to converge in every cycle. The complexities and the run-times of BA and IBU are summarised in Table~\ref{table:runtime}. In both cases, we assigned the value of the loss parameter signifying the QoS of the users, $\beta$, to be $0.5$ and $1$. This was done to test the performance of PRIVIC in estimating the true PMF under two different levels of privacy. Each experiment was run for 5 rounds of simulation to calibrate the randomness of the sampling and obfuscation. In each cycle of PRIVIC, across all the settings, BA was initiated with the uniform marginal $c_0$ and a uniform distribution over the space of locations as the ``starting guess'' of the true distribution.

With $\beta=1$, BA produces a geo-indistinguishable mechanism that injects less local noise than that obtained with $\beta=0.5$. As a result, PRIVIC obtains a more accurate estimate of the true PMF for the $\beta=1$ than for $\beta=0.5$. However, in both cases, the EMD between the true and the estimated distributions is very low, indicating that the PRIVIC mechanism is able to preserve a good level of statistical utility. Moreover, for both Paris and San Francisco, PRIVIC seems to significantly improve its estimation of the true PMF with every iteration until it converges to the MLE. Comparing Figures \ref{fig:Estimated_Density} and \ref{fig:density}, we see that the estimations of the true distributions of the locations in Paris and San Francisco by IBU under PRIVIC for both the settings of the loss parameter are fairly accurate. However, as we would anticipate, the statistical utility for $\beta=1$ is better than that for $\beta=0.5$.   

\begin{figure}[htbp]
\centering
\begin{subfigure}[b]{0.493\columnwidth}
   \includegraphics[width=1.02\columnwidth]{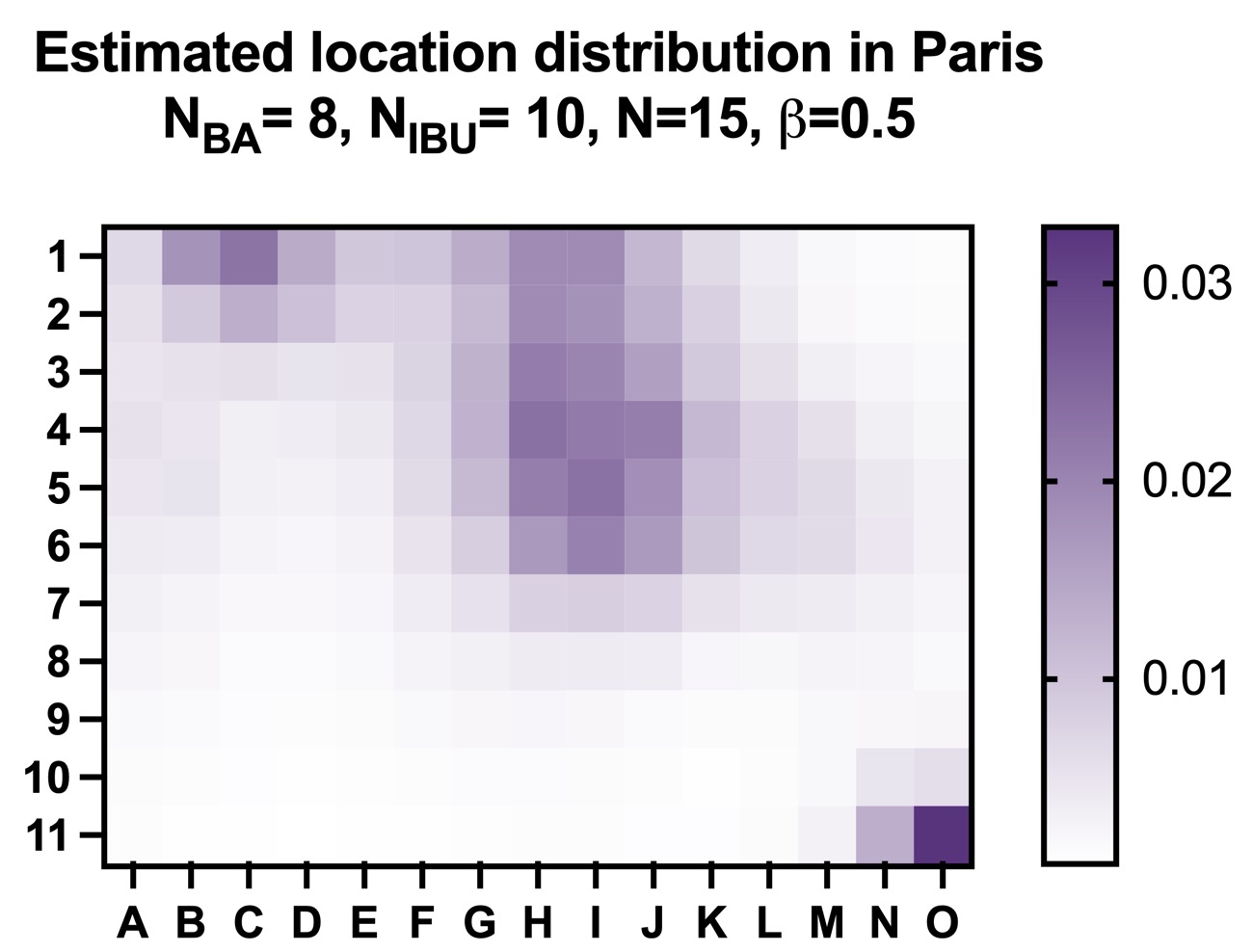}
   \caption{Paris; $\beta=0.5$}
   \label{fig:ParisEstimate_Beta:0_5} 
\end{subfigure}
\begin{subfigure}[b]{0.493\columnwidth}
   \includegraphics[width=1.02\columnwidth]{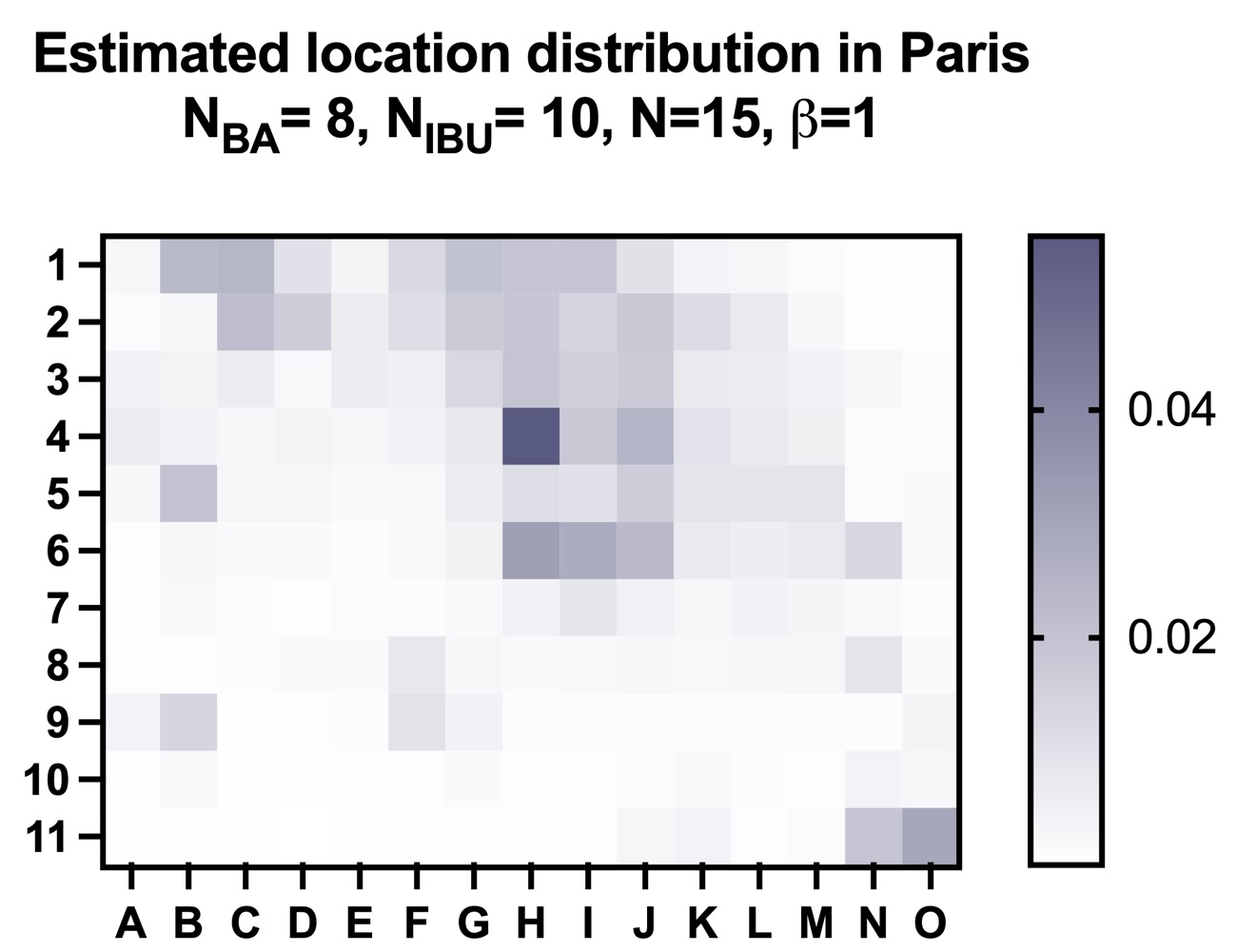}
   \caption{Paris; $\beta=1$}
   \label{fig:ParisEstimate_Beta:1}
\end{subfigure}
\begin{subfigure}[b]{0.493\columnwidth}
   \includegraphics[width=1.02\columnwidth]{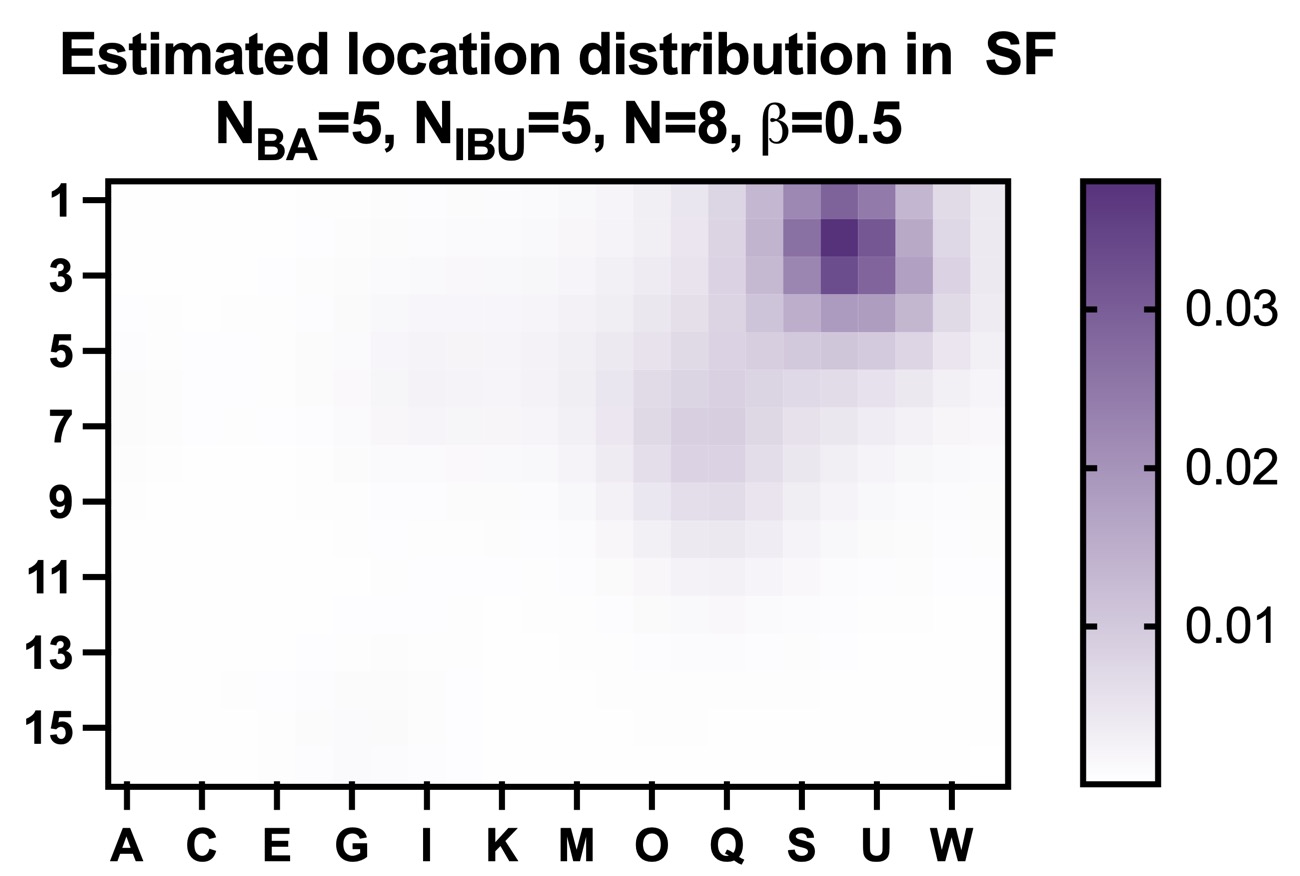}
   \caption{San Francisco; $\beta=0.5$}
   \label{fig:SFEstimate_Beta:0_5} 
\end{subfigure}
\begin{subfigure}[b]{0.493\columnwidth}
   \includegraphics[width=1.02\columnwidth]{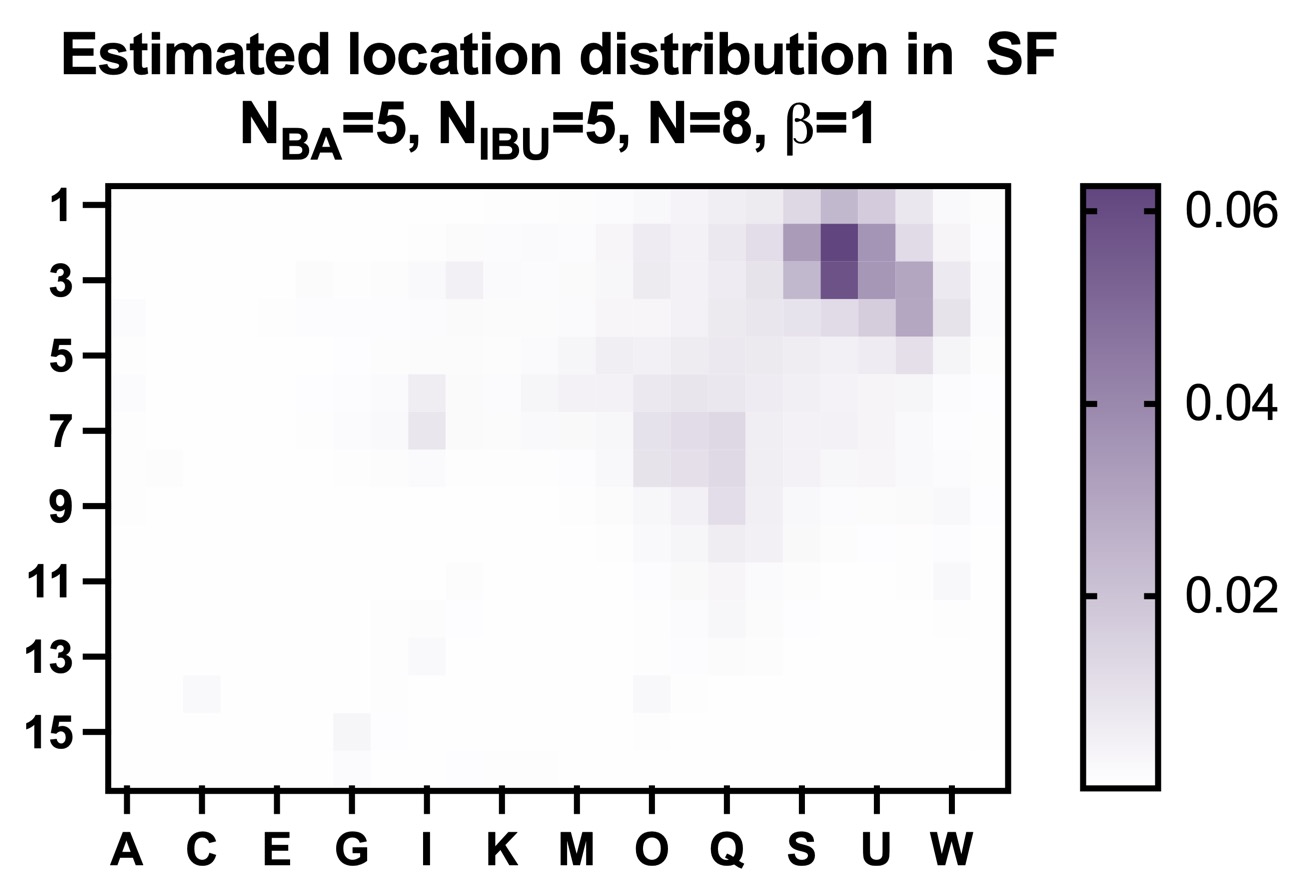}
   \caption{San Francisco; $\beta=1$}
   \label{fig:SFEstimate_Beta:1}
\end{subfigure}
\caption{Visualization of the estimated true distribution of the locations in Paris ((a) and (b)) and San Francisco ((c) and (d)) by PRIVIC after its convergence; the first column is for $\beta=0.5$ and the second column is for $\beta=1$.}
\label{fig:Estimated_Density}
\end{figure}

\begin{figure}[htbp]
\centering
\begin{subfigure}[b]{0.45\columnwidth}
   \includegraphics[width=1.15\columnwidth]{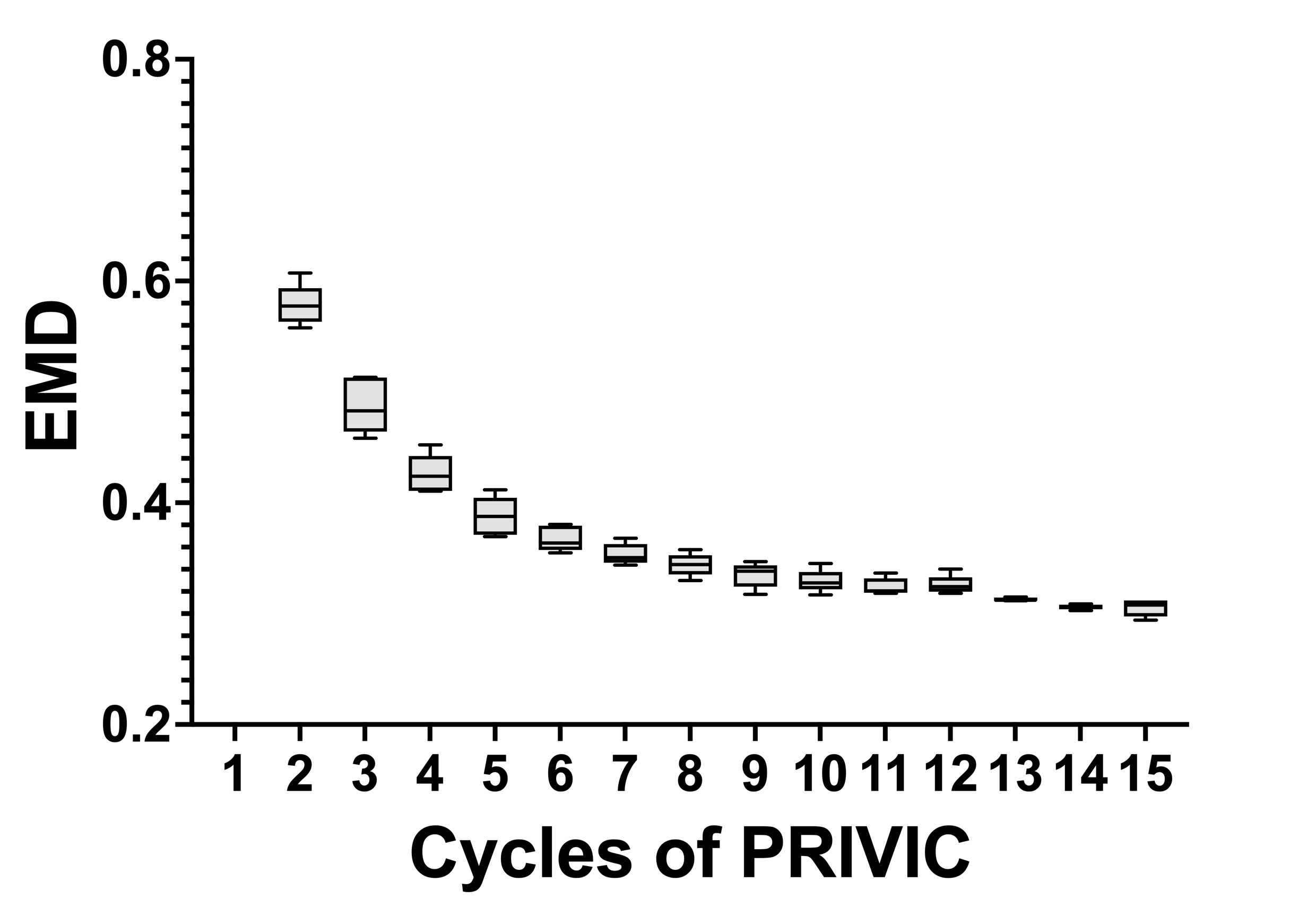}
   \caption{$\beta=0.5$}
   \label{fig:UtilityParisZoomedB1} 
\end{subfigure}
\begin{subfigure}[b]{0.45\columnwidth}
   \includegraphics[width=1.15\columnwidth]{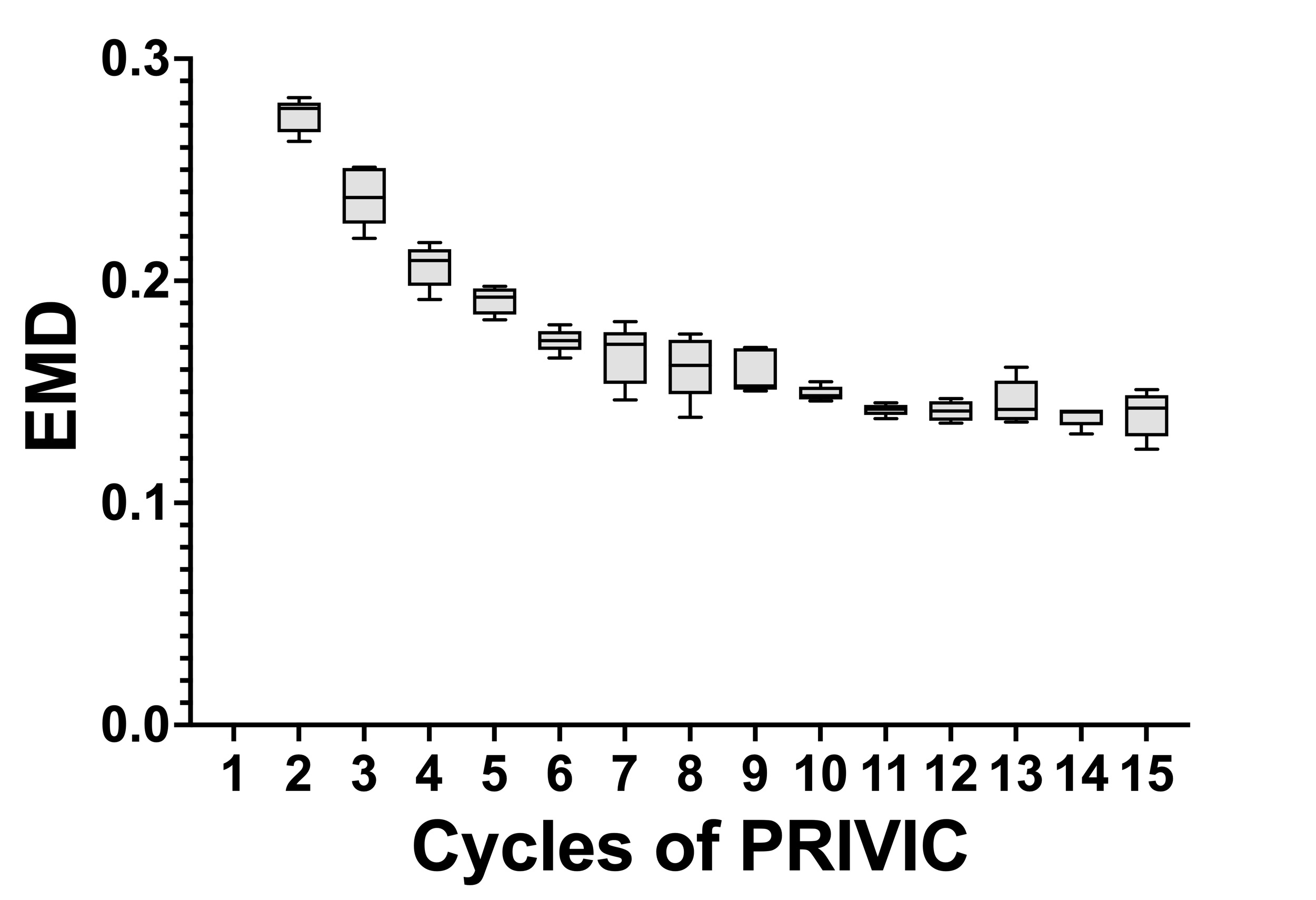}
   \caption{$\beta=1$}
   \label{fig:UtilityParisZoomedB0_5}
\end{subfigure}
\caption{(a) and (b) show the EMD between the true PMF of the Paris locations and its estimation by PRIVIC in each of its cycle for $\beta=0.5$ and $\beta=1$, respectively.}
\label{fig:UtilityParis}
\end{figure}

\begin{figure}[htbp]
\centering
\begin{subfigure}[b]{0.45\columnwidth}
   \includegraphics[width=1.1\columnwidth]{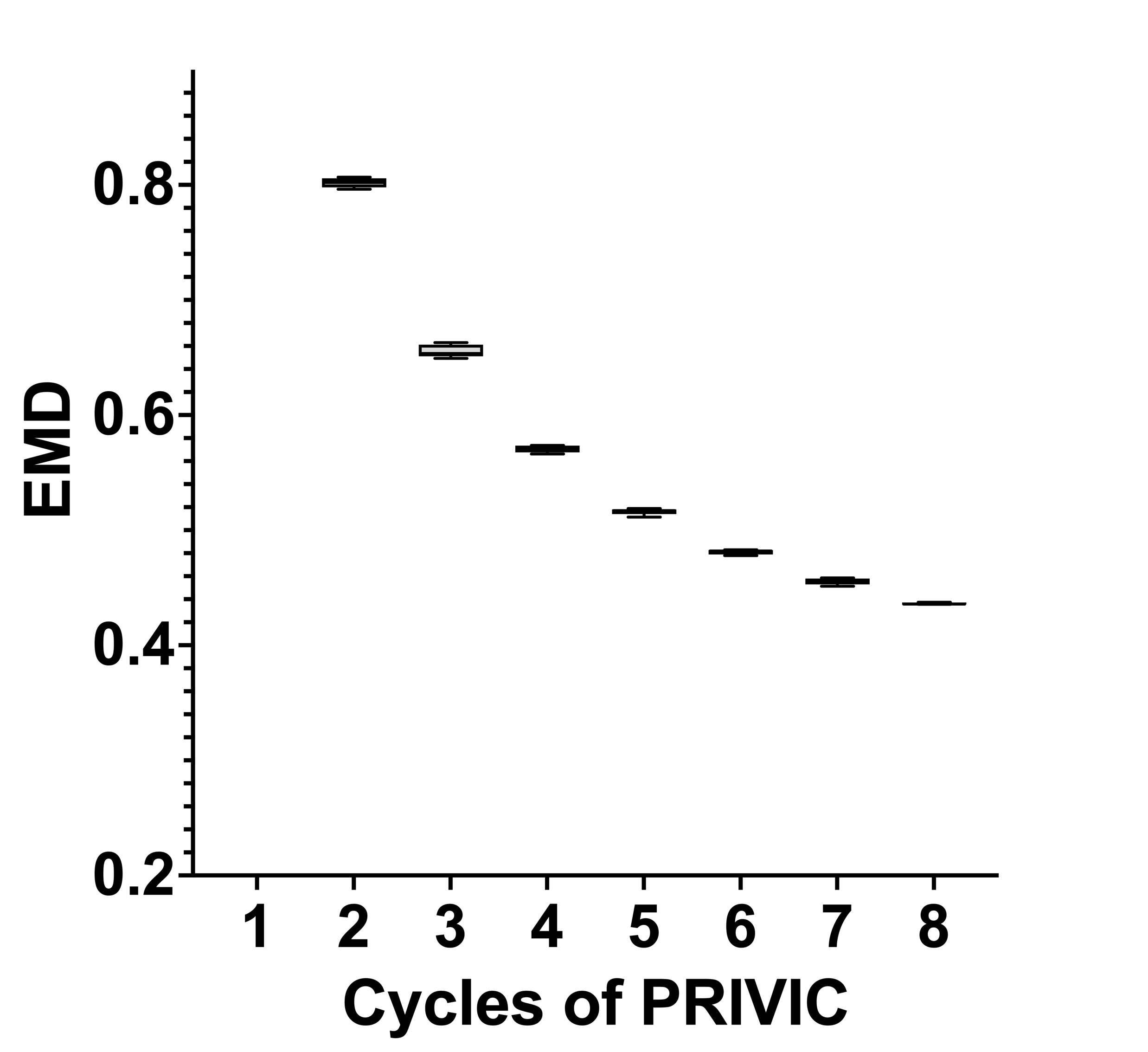}
   \caption{$\beta=0.5$}
   \label{fig:UtilitySFZoomedB0_5} 
\end{subfigure}
\begin{subfigure}[b]{0.45\columnwidth}
   \includegraphics[width=1.1\columnwidth]{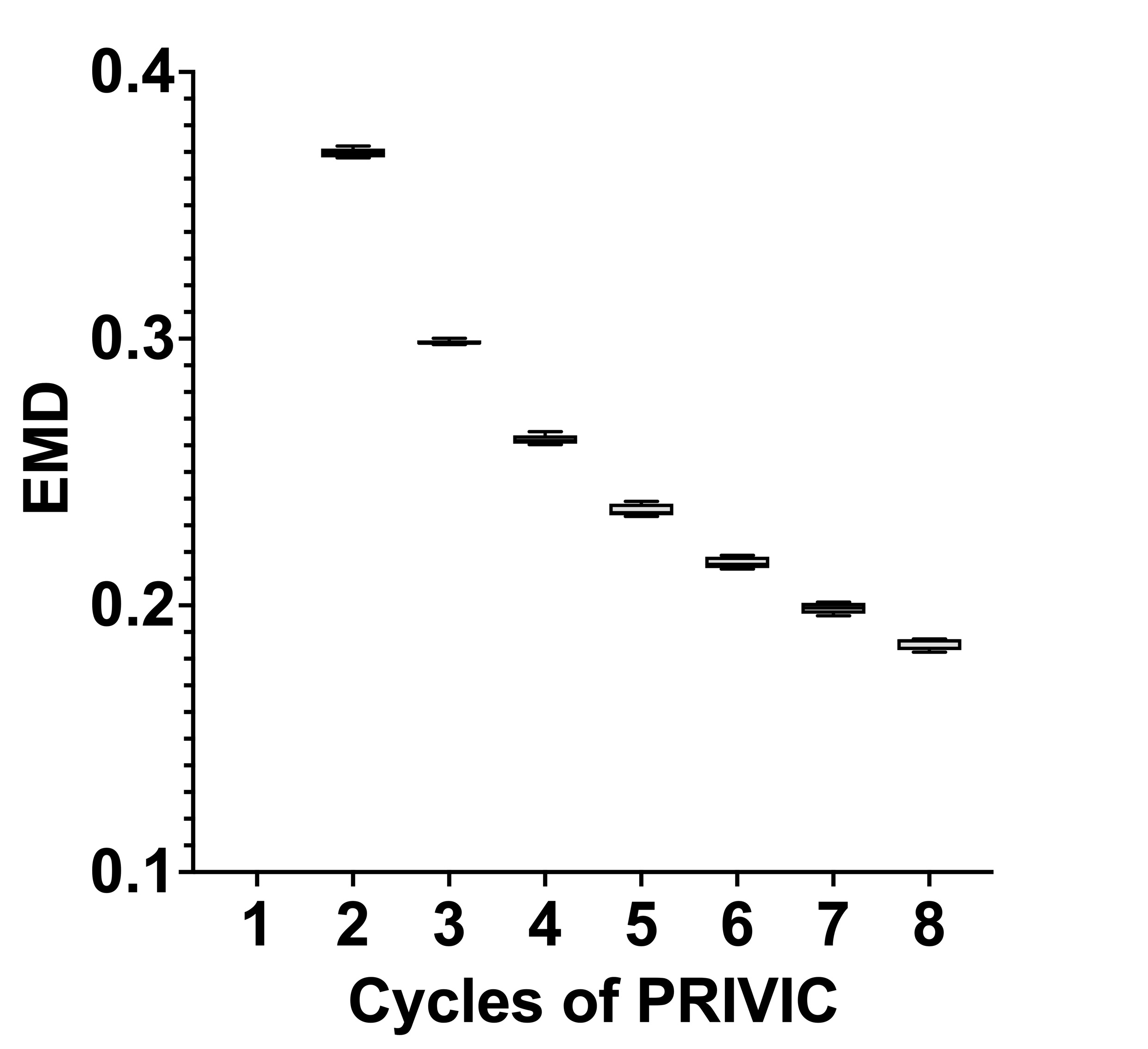}
   \caption{$\beta=1$}
   \label{fig:UtilitySFZoomedB1}
\end{subfigure}
\caption{(a) and (b) show the EMD between the true PMF of the San Francisco locations and its estimation by PRIVIC in each of its cycles for $\beta=0.5$ and $\beta=1$, respectively.}
\label{fig:UtilitySF}
\end{figure}

Now we shift our attention to analyze the performance of PRIVIC in preserving the statistical utility and its long-term behaviour of the two datasets. Figure~\ref{fig:UtilityParis} shows us the EMD between the true distribution of the locations in Paris and its estimate by IBU under PRIVIC in each of its 15 cycles under the two settings of privacy ($\beta=0.5,\,1$). One of the most crucial observations here is that the EMD between the true and the estimated PMFs seems to decrease with the number of iterations and it finally converges, implying that the estimation of PMFs given by PRIVIC seems to improve at the end of each cycle and, eventually, it converges to the MLE of the prior of the noisy locations, giving the estimate of the true PMF. This, empirically, suggests the convergence of the entire method. This is a major difference from the work of \cite{Oya:19:EuroSnP} which, as we pointed out before, has the potential of encountering an LPPM which is optimal according to the standards set by Shokri et al. in \cite{ShokriQuantifyingLocPriv2011} but the EM method used to estimate the true distribution would fail to converge for that mechanism as illustrated in Example~\ref{examp:BadLPPM}. We observe a very similar trend for the San Francisco dataset. Figure~\ref{fig:UtilitySF} shows the statistical utility of the mechanism generated by PRIVIC under each of its 8 cycles for $\beta=0.5$ and $\beta=1$. The explicit values of the EMD between the true and the estimated PMFs on the location data from Paris and San Francisco for both the settings of the loss parameter can be found in Tables \ref{table:UtilityData_Paris} and \ref{table:UtilityData_SF} in Appendix~\ref{app:tables}.

In the next part of the experiments, we set ourselves to dissect the trend of the statistical utility harboured by PRIVIC w.r.t. the level of geo-ind it guarantees. We recall that the higher the value of $\beta$, the lesser the local noise that is injected into the data, and, hence, the worse will be the statistical utility, staying consistent with our observations in Figure~\ref{fig:Estimated_Density}. We continue working with the location data from Paris and San Francisco obtained from the Gowalla dataset in the same framework as described before. We consider $\beta$ taking the values $0.1,0.3,0.5,0.7,0.9,1$, and for each value of the loss parameter, we run PRIVIC on both datasets using the same number of iterations as in the previous experiments. We adhere to 5 rounds of simulation for each $\beta$ to account for the randomness generated in the obfuscation process.

\begin{figure}[htbp]
\centering
\begin{subfigure}[b]{0.49\columnwidth}
   \includegraphics[width=0.9\columnwidth]{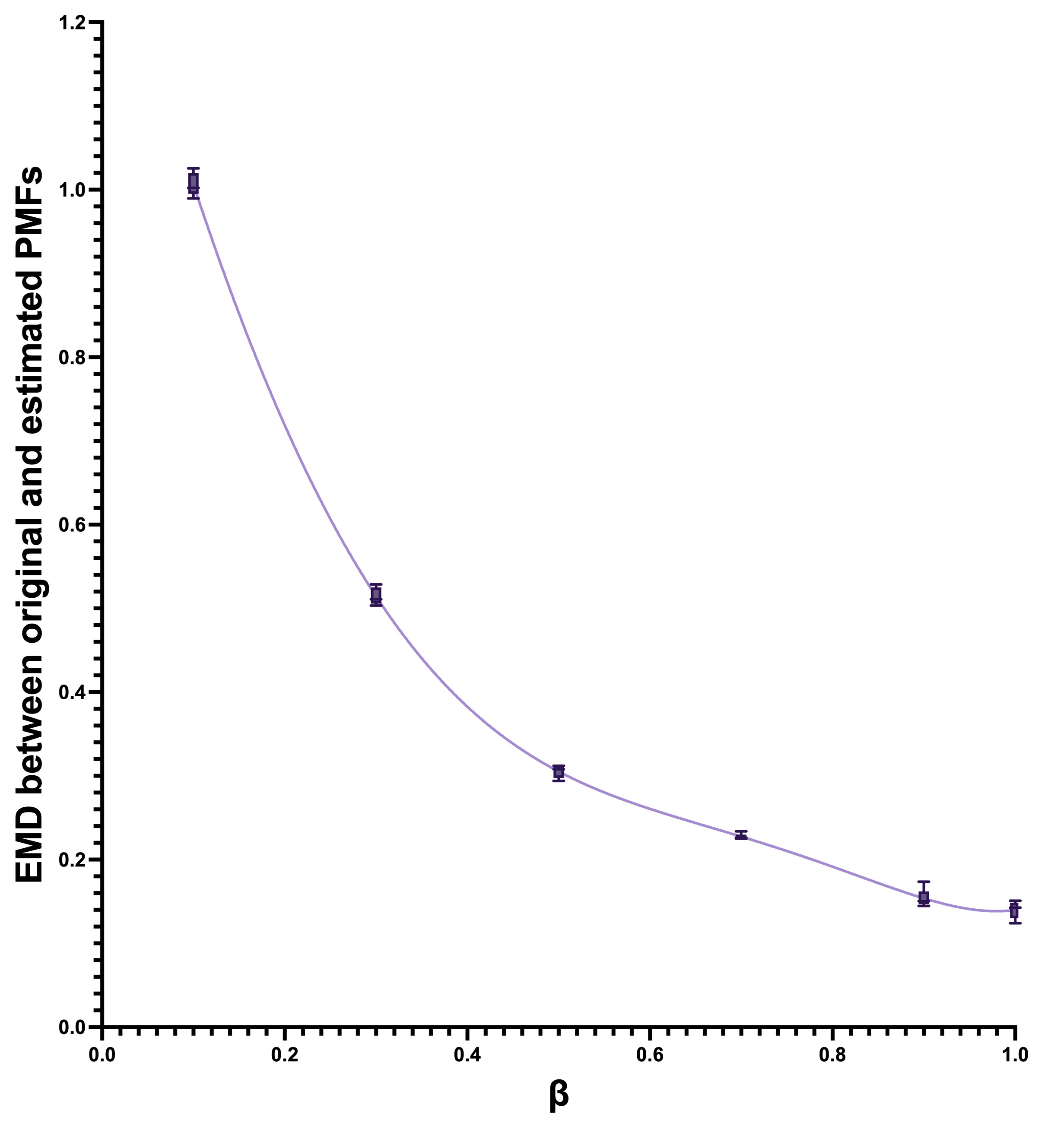}
   \caption{Paris}
   \label{fig:UtilityVsDistortion_Paris} 
\end{subfigure}
\begin{subfigure}[b]{0.49\columnwidth}
   \includegraphics[width=0.9\columnwidth]{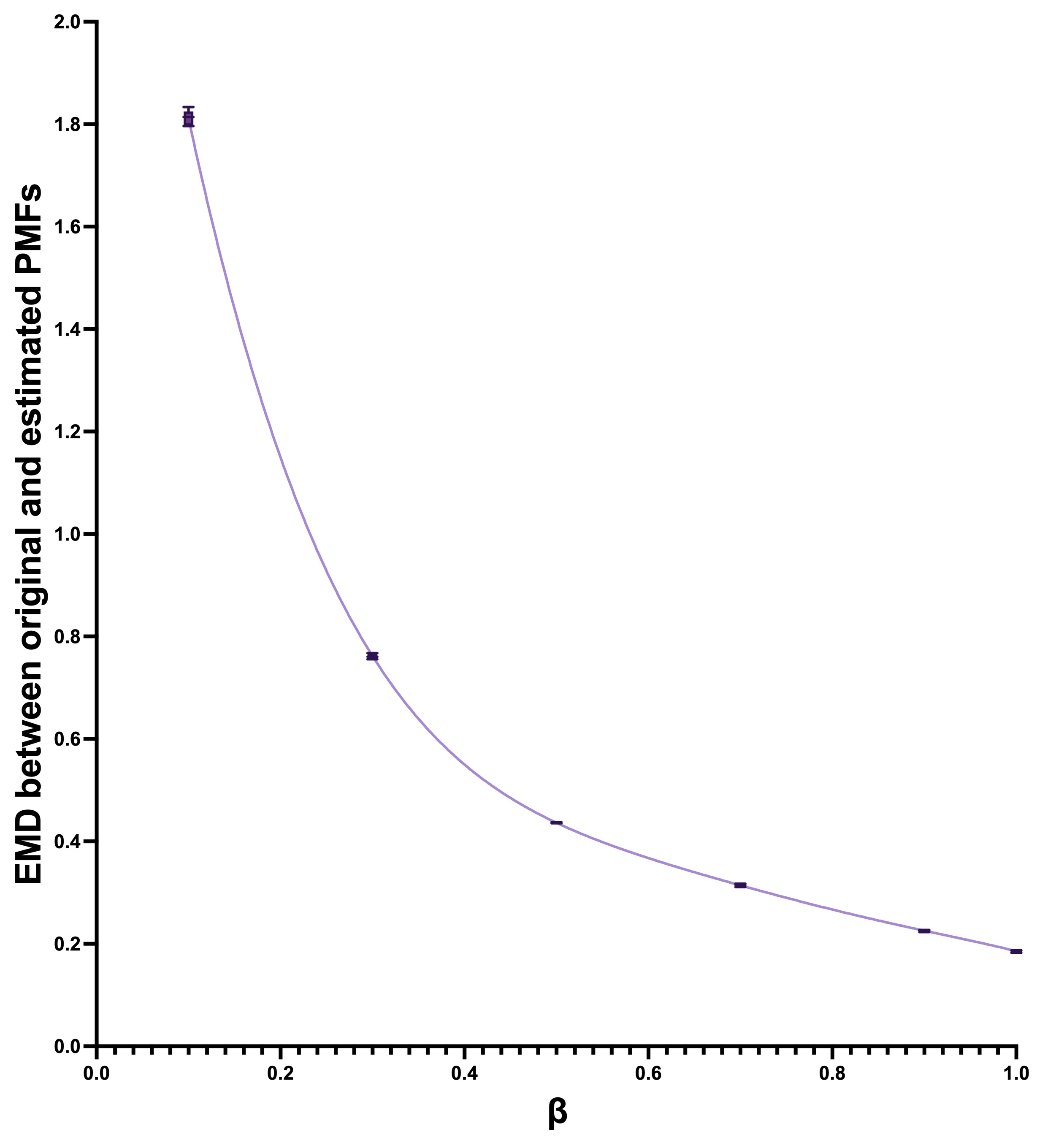}
   \caption{San Francisco}
   \label{fig:UtilityVsDistortion_SF}
\end{subfigure}
\caption{(a) and (b) illustrate that EMD between the true and the estimated distributions of the locations in Paris and San Francisco, respectively, after the empirical convergence of PRIVIC for the different values of the loss parameters $\beta$.}
\label{fig:UtilityVsDistortion}
\end{figure}

Figure~\ref{fig:UtilityVsDistortion} shows us that the difference between the true and the estimated PMFs under PRIVIC starts by sharply decreasing and then eventually stabilizes with an increase in the value of the loss parameter. In other words, as the intensity of the local noise decreases, we will end up estimating the unique MLE of the original distribution while optimizing MI and the users' QoS. Both the location datasets result in a Pareto curve showing a similar trend. This depicts an improvement of the estimated PMF until it converges to the true distribution. This observation complements the Pareto-optimality of MI with the maximum average distortion as studied in \emph{rate-distortion theory}~\cite{ShannonInfoTheory}, and thus, we empirically weave together the two ends of utility with the information theoretical notion of privacy under PRIVIC. 

\subsubsection*{Discussion}

As a justification for the applicability and the working of our method, in a setting where the service providers periodically collect location data from clients, it is reasonable to assume that, over time, they would like to maximize their utility by accurately approximating the true distribution of the population for improving their service in various aspects (crowd management, security enhancement, WLAN hotspot positioning, etc.). BA, in addition to guaranteeing geo-ind, acts as an elastic location-privacy mechanism and optimizes between MI and the data owners' QoS when it initiates with the true prior. Therefore, as every iteration of PRIVIC improves the estimation of the original distribution, as seen in Figures \ref{fig:StatUtilBAvsLapParis} and \ref{fig:StatUtilBAvsLapSF}, which is used as the starting distribution in its next cycle, the overall privacy protection and its trade-off with QoS of the users will also improve, motivating the users and the service providers comply with PRIVIC to act in their best interests and, in turn, engaging them in a positive feedback loop to maximize the corresponding privacy and utility goals. 

\revision{
\section{Vulnerability of PRIVIC}~\label{sec:vul_PRIVIC}
\vspace{-1cm}
\begin{figure}[htbp]
\centering
   \begin{subfigure}[b]{1\columnwidth}
   \includegraphics[width=0.9\columnwidth]{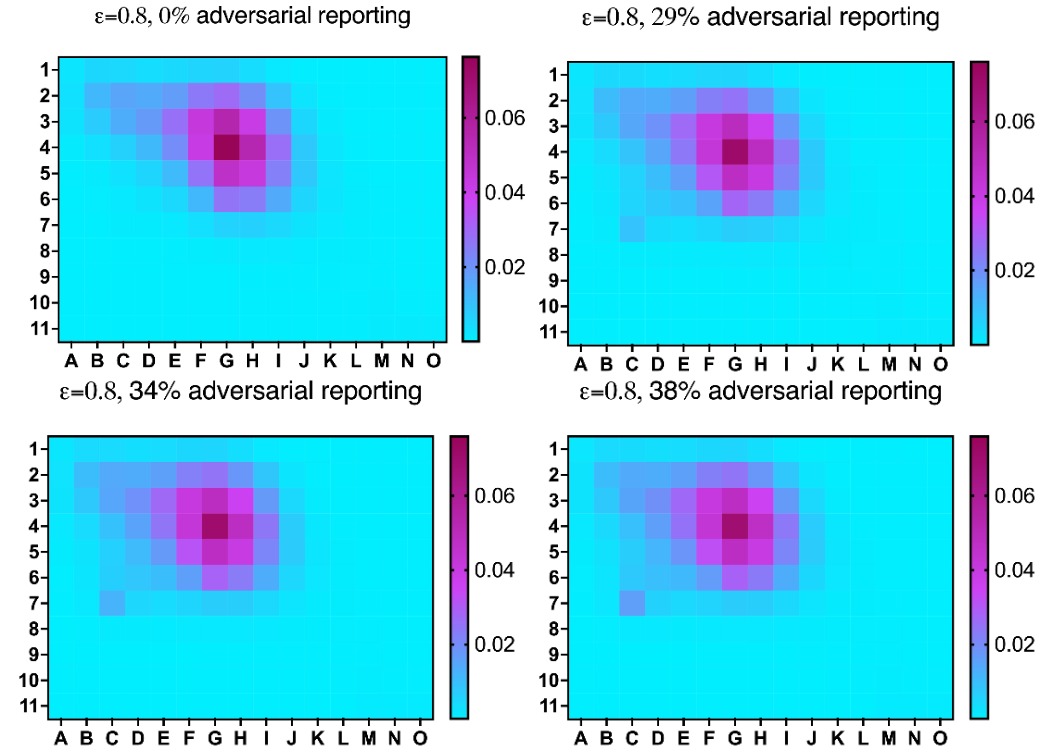}
   \caption{\revision{Obfuscation distribution of a vulnerable location in the map using BA satisfying $0.8$-geo-ind with (clockwise) $0\%$, $29\%$, $34\%$, and $39\%$ adversarial users, respectively.}}
   \label{fig:adv_input_B04} 
\end{subfigure}
\begin{subfigure}[b]{1\columnwidth}
   \includegraphics[width=0.9\columnwidth]{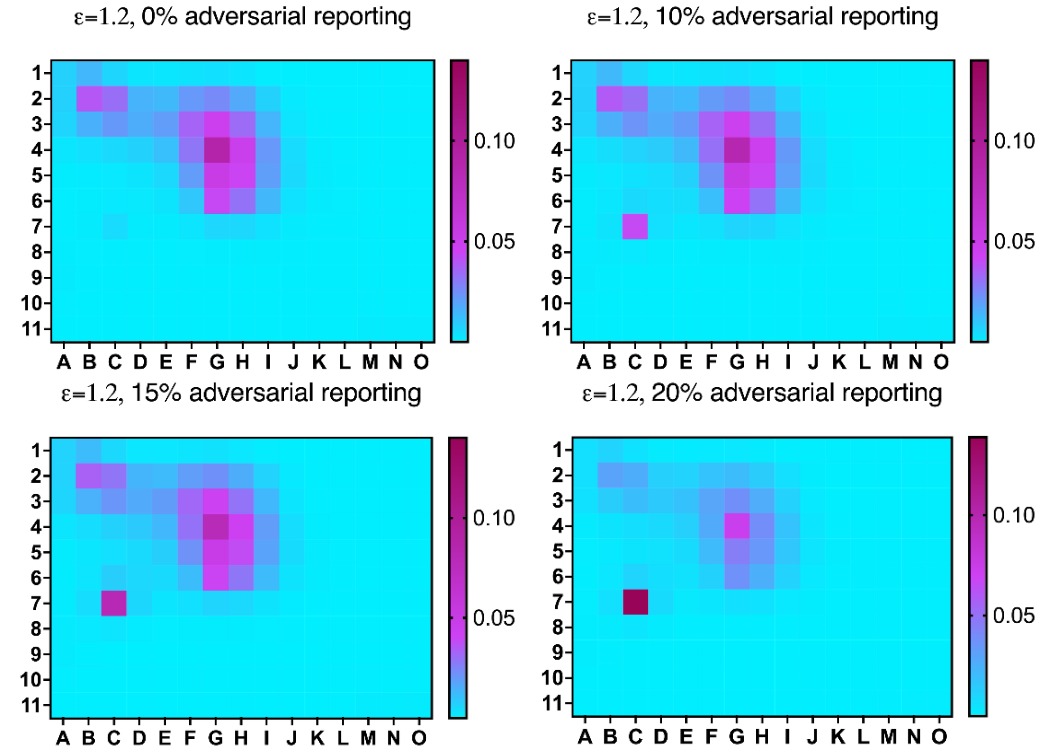}
   \caption{\revision{Obfuscation distribution of the vulnerable location in the map using BA satisfying $1.2$-geo-ind with (clockwise) $0\%$, $10\%$, $15\%$, and $20\%$ adversarial users, respectively.}}
   \label{fig:adv_input_B06}
\end{subfigure}
\caption{\revision{Effect on the privacy provided by BA to obfuscate the geo-spatially isolated location (Location A as in Figure~\ref{fig:str_vul}) for different fractions of adversarial users who intentionally report their locations falsely under two different levels of formal geo-ind guarantees by BA.}}
\label{fig:adv_input}
\end{figure}}

\color{black}
In this section, we illustrate a potential vulnerability of PRIVIC when a subset of colluded users (\emph{adversarial users}) intentionally deviate from the correct use of the protocol. 

The attack consists in falsely  reporting their location in order to alter the estimation of the true distribution and, consequently, the obfuscation mechanism produced by BA. Specifically, we study two cases: (i) adversaries reporting a crowded location (\emph{strong location}) and  adversaries reporting an isolated location (\emph{vulnerable location}). 

We used the real locations from the Paris dataset with the geo-spatially isolated ``island'' (as illustrated by~Figure~\ref{fig:str_vul} presented in Section~\ref{sec:LPPM_with_BA}) representing a strong and a vulnerable location in the map denoted by points A and B in the figure, respectively. We performed the experiments with two different levels of formal geo-indistinguishability ($\epsilon=0.8$ and $\epsilon=1.2$) considering different fractions of ``adversarial data submissions'' (i.e., adversarial users reporting their locations falsely to compromise the privacy of other users) in each case. 

Although, in both cases (i) and (ii), we observed that with an increase in the fraction of adversaries, the probability mass assigned by BA (used to obfuscate the corresponding points locally) becomes higher in and around the corresponding reported points, the impact of privacy differs across both settings. For (i), the obfuscation distribution happens to be weighed heavily around the true crowded location (point B) by both BA and LAP (as illustrated by Figures \ref{fig:Elastic_str} and \ref{fig:Elastic_vul}) even without any adversarial users. This trend was seen to continue even when we assumed different levels of adversaries.

However, case (ii) represents a much more serious attack. As the number of adversarial users increases, BA and, in turn, PRIVIC become less potent to be able to protect the privacy of  honest users who are genuinely located in an isolated location on the map. We also observe that BA and PRIVIC start behaving more like LAP. In particular, Figures \ref{fig:adv_input_B04} and \ref{fig:adv_input_B06} illustrate that the obfuscation distribution generated by BA satisfying geo-ind with $\epsilon=0.8$ and $\epsilon=1.2$, respectively, of the (non-adversarial) users located in point A assigns more and more weight to and around point A which, as a result, makes them more and more identifiable. This evaluation of the vulnerability of PRIVIC under adversarial data submission essentially exposes a weakness of the elastic distinguishability metric. We plan to address this aspect and aim to make PRIVIC more robust against adversarial users in our future works.

\color{black}
\section{Conclusion}~\label{sec:conclusion}

We have bridged some  ideas from information theory and statistics to develop a method allowing an incremental collection of location data while protecting the privacy of the data owners, upholding their quality of service, and preserving the statistical utility for the data consumers. Specifically, we have proposed the Blahut-Arimoto algorithm as a location-privacy mechanism, showing its extensive privacy-preserving properties and its other advantages over the state-of-the-art Laplace mechanism for geo-ind. Further, we have exhibited its duality with the iterative Bayesian update and explored this connection to present an iterative method (PRIVIC) for incremental collection of location data with formal guarantees of geo-ind and an elastic distinguishability metric, while optimizing the QoS of the users and their privacy from an information theoretical perspective. Moreover, PRIVIC efficiently estimates the MLE of the distribution of the original data and, thus, yields a high statistical utility for the service providers.  Finally, we have illustrated the convergence and the general functioning of PRIVIC with experiments on real location datasets. We believe that our results can be extended easily to other kinds of data, including those with high dimensions, and to other notions of distortion measures since the analysis carried out in this paper does not depend on the notion of distance used.


\bibliographystyle{IEEEtran}
\bibliography{references}
\appendices
\section{Proofs}~\label{app:proofs}
 
\BAElastic*
\begin{proof}
Let $x,y\in\mathcal{X}$ be any true and reported location, respectively. Letting $\hat{\mathcal{C}}_{\text{BA}}$ to be the limiting mechanism generated by BA, to show that $\hat{\mathcal{C}}_{\text{BA}}$ possesses an elastic distinguishability metric, we need to ensure that:
\begin{enumerate}
    \item The probability of reporting $y$ to obfuscate $x$ given by $\hat{\mathcal{C}}_{\text{BA}}$ should be exponentially reducing w.r.t. the Euclidean distance between $x$ and $y$, staying consistent with the essence of geo-ind (the property captured by \eqref{eq:elasticprop1}).
    \item Under $\hat{\mathcal{C}}_{\text{BA}}$, the probability of reporting $y$ to obfuscate $x$ should be taking into account the mass of reported points around $y$, i.e., the more geo-spatially isolated (from other reported points) $y$ is in the space, the less likely it should be to report it, as, ideally, we would like to have $x$ being reported as a location amidst a crowd of other reported locations  (the property captured by \eqref{eq:elasticprop2}).
\end{enumerate}

Let's simplify the notation and denote $\mathbb{P}[\hat{\mathcal{C}}_{\text{BA}}(x)=y]$ as $\mathbb{P}_{{\text{BA}}}[y|x]$ and let $q(y)$ be the probability mass of the observed location $y$. Hence, for being an elastic location-privacy mechanism, $\hat{\mathcal{C}}_{\text{BA}}$ should satisfy \eqref{eq:elasticprop1} and \eqref{eq:elasticprop2}, i.e., we must have:
\begin{align}
    \mathbb{P}_{\text{BA}}[y|x]\propto \exp{-\beta d_{\text{E}}\left(x,y\right)}~\label{eq:BAelasticprop1}\\
    \mathbb{P}_{\text{BA}}[y|x]\propto q(y)~\label{eq:BAelasticprop2}
\end{align}
Therefore, in order to satisfy \eqref{eq:BAelasticprop1} and \eqref{eq:BAelasticprop2}, it is sufficient to have:
\begin{align}
    \mathbb{P}_{\text{BA}}[y|x]\propto \exp{-\beta d_{\text{E}}\left(x,y\right)+\ln{q(y)}}\nonumber\\
    \implies \mathbb{P}_{\text{BA}}[y|x]\propto q(y)\exp{-\beta d_{\text{E}}\left(x,y\right)}\nonumber\nonumber\\
    =\frac{q(y)\exp{\beta d_{\text{E}}(x,y)}}{\sum\limits_{z\in\mathcal{X}}q(z)\exp{-\beta d_{\text{E}}(x,z)}}~\label{eq:elasticBA}
\end{align}
Now, it's sufficient to note that, if we interpret the mass of $y$  as the probability of being reported by the mechanism, \eqref{eq:elasticBA} is exactly the fixpoint of $\mathcal{G}\circ\mathcal{F}$, cf. Remarks \ref{rem:BATransfom} and \ref{rem:BAConverge}. 
\end{proof}

\BAisInvertible*
\begin{proof}
For notational convenience, in this proof, we shall denote the Euclidean distance $d_{\text{E}}(.)$ as $d(.)$. For any $t\geq 1$, let $\mathcal{C}^{(t)}$ be the channel generated at the $t$'th iteration of BA. Hence, we have:
\begin{align}
    \mathcal{C}^{(t)}_{x,y}=\frac{c_{t-1}(y)\exp{-\beta d(x,y)}}{\sum_{z\in\mathcal{Y}}c_t(z) \exp{-\beta d(x,z)}}
\end{align}
Let $\mathcal{C}^{'}\in \vb*{C}(\mathcal{X},\mathcal{X})$ such that $\mathcal{C}^{'}_{x,y}=\exp{-\beta d(x,y)}$. Correspondingly, let us define $\mathcal{C}^{''(t)}, \mathcal{C}^{'''(t)}\in \vb*{C}(\mathcal{X},\mathcal{X})$ s.t. $\mathcal{C}^{''(t)}_{x,y}=c_{t-1}(y)\mathcal{C}^{'}_{x,y}$ and $\mathcal{C}^{'''(t)}_{x,y}=K_x\mathcal{C}^{''(t)}_{x,y}$ where $K_x=(\sum_{z\in\mathcal{Y}}c_t(z) \exp{-\beta d(x,z)})^{-1}$. Therefore, we have $\mathcal{C}^{'''(t)}=\mathcal{C}^{(t)}$. 

Exploiting the fact that scaling of rows and columns of matrices by real numbers (elementary operations on rows and columns) preserves their linear independence, we ensure that if $\mathcal{C}^{'}$ is invertible, then so is $\mathcal{C}^{''(t)}$ (elementary column operation on $\mathcal{C}^{'}$) which, in turn, implies that $\mathcal{C}^{'''(t)}=\mathcal{C}^{(t)}$ is invertible (elementary row operation on $\mathcal{C}^{''(t)}$). Therefore, in order to show $\mathcal{C}^{(t)}$ is invertible, it is sufficient to prove that $\mathcal{C}^{'}$ is invertible. 

Note that $\exp{-\beta d(x,y)^2}=\exp{-\beta ||x-y||^2_2}$ is the \emph{Gaussian kernel} for any $\beta>0$ and is positive definite~\cite{Haussler1999ConvolutionKO,Hofmann_2008_kernel}. Furthermore, Schoenberg~\cite{Schoenberg} observed that for any \emph{completely monotone function} $g\colon\mathbb{R}_{\geq 0}\mapsto \mathbb{R}$, we can use \emph{Hausdorff–Bernstein–Widder theorem}~\cite{Bernstein,widder1941laplace} to deduce that \emph{radial basis function (RBF)} kernels such as $\exp{-\beta g\left(||x-y||^2_2\right)}$ are also positive definite. Moreover, in addition to being positive definite, it was also shown that Gaussian kernels are \emph{strictly} positive definite~\cite{wendland_2004, Hofmann_2008_kernel}. 

Let $f\colon \mathbb{R}_{\geq 0}\mapsto\mathbb{R}$ be the \emph{square-root} function, i.e., $f(x)=\sqrt{x}$ for all $x\in\mathbb{R}_{\geq 0}$. Therefore, observing that $f$ is completely monotone and recalling that Gaussian kernels are strictly positive definite, i.e., $\vb*{z}^{\operatorname{T}} \exp{-\beta ||x_i-x_j||_2^2}\vb*{z}\geq 0$ for every $\vb*{z}\in\mathbb{R}^m$ with equality holding iff $\vb*{z}=\vb*{0}$, we can use \emph{Schoenberg theorem}~\cite{Schoenberg} to show the strict positive definiteness of $\exp{-\beta f\left(d(x,y)^2\right)}=\exp{-\beta d(x,y)}$. Hence, noting that $\mathcal{C}'$ is the Gram matrix of the RBF kernel $\exp{-\beta d(x,y)}$, we must have $\vb*{z}^{\operatorname{T}} \mathcal{C}'\vb*{z}\geq 0$ for every $\vb*{z}\in\mathbb{R}^m$ with equality holding iff $\vb*{z}=\vb*{0}$. This implies that $\mathcal{C}'$ is positive definite and, hence, invertible. Therefore, in turn, $\mathcal{C}^{(t)}$ is invertible.

\end{proof}
\color{black}
\PRIVICConverges*
\begin{proof}
For $1\leq t \leq N$ and $1\leq i\leq n$, in the $t$'th round of PRIVIC, the $i$'th sampled user locally sanitize their location $x^{(t)}_i$ with $\hat{\mathcal{C}}^{(t)}$ and reports the noisy location $y^{(t)}_i$. Therefore, the \emph{combined mechanism} (referred to as \emph{output probability matrix} in \cite{EhabGIBU}) for implementing GIBU is $\mathcal{G}=\begin{pmatrix}
        \hat{\mathcal{C}}^{(1)}&\ldots & \hat{\mathcal{C}}^{(N)}
    \end{pmatrix}$ s.t.

\begin{align}
    &\mathcal{G}\left(x,y^{(t)}_i\right)=\mathbb{P}\left[\left.y^{(t)}_i\right| x\right]=\hat{\mathcal{C}}^{(t)}\left(x,y^{(t)}_i\right)\nonumber\\
    &\forall x\in\mathcal{X},\,i\in\{1,\ldots,n\}\nonumber.
\end{align}
By Theorem~\ref{th:BA_is_invertible}, $\hat{\mathcal{C}}^{(t)}$ is invertible for every $t\geq 1$ and let $\hat{\mathcal{C}}^{(t)\,-1}$ denote the corresponding inverse of $\hat{\mathcal{C}}^{(t)}$. Therefore, defining $\mathcal{G}'$ s.t $\mathcal{G}'=
\frac{1}{N}\begin{pmatrix}
        \hat{\mathcal{C}}^{(1)\,-1}&\ldots & \hat{\mathcal{C}}^{(N)\,-1}
    \end{pmatrix}^{\operatorname{T}}$ ensures that $\mathcal{G}\cdot \mathcal{G}'=\mathbb{I}_{m}$ where $m=|\mathcal{X}|$. Therefore, $\mathcal{G}$ is right-invertible.
    
    Hence, combining the right-invertibility of $\mathcal{G}$ with Theorem 3 (GIBU converges to MLEs) and Corollary 1 (right-invertibility of the combined channel of GIBU implies unique MLE) of \cite{EhabGIBU}, we can conclude that $\texttt{GIBU}\left(\left(\hat{\mathcal{C}}^{(1)},\vb*{y}^{(1)}\right),\ldots, \left(\hat{\mathcal{C}}^{(N)},\vb*{y}^{(N)}\right)\right)$ estimates the unique MLE of the prior $\pi_{\mathcal{X}}$, implying that PRIVIC converges. 
\end{proof}

\color{black}

\newpage
\section{Tables}~\label{app:tables}
\begin{table}[htbp]
\centering
\caption{EMD between the true and the estimated PMFs by PRIVIC on the Paris locations.}~\label{table:UtilityData_Paris}
\resizebox{1\columnwidth}{!}{%
  \begin{tabular}{cccccc|ccccc}
    \hline
    \multirow{2}{*}{$N$} & 
      \multicolumn{5}{c|}{$\beta=1$} & 
      \multicolumn{5}{c}{$\beta=0.5$} \\
    & Round 1 & Round 2 & Round 3 & Round 4  & Round 5 & Round 1 & Round 2 & Round 3 & Round 4 & Round 5 \\
    \hline
    1 & 2.02262 & 2.02262 & 2.02262 & 2.02262 & 2.02262 & 2.02262 & 2.02262 & 2.02262 & 2.02262 & 2.02262\\
    \hline
    2 & 0.27104 & 0.27796 & 0.28247 & 0.27758 & 0.26276 & 0.57738 & 0.57994 & 0.55791 & 0.60717 & 0.56880\\
    \hline
    3 & 0.21916 & 0.23750 & 0.25035 & 0.25116 & 0.23241 & 0.51324 & 0.48295 & 0.47043 & 0.51285 & 0.45826\\
    \hline
    4 & 0.19156 & 0.21115 & 0.21726 & 0.20913 & 0.20408 & 0.43184 & 0.42398 & 0.41040 & 0.45230 & 0.41119\\
    \hline
    5 & 0.18241 & 0.19264 & 0.19570 & 0.19747 & 0.18728 & 0.39741 & 0.38771 & 0.37284 & 0.41176 & 0.36953\\
    \hline
    6 & 0.16526 & 0.18020 & 0.174578 & 0.17310 & 0.17268 & 0.37818 & 0.35482 & 0.36039 & 0.36375 & 0.38045\\
    \hline
    7 & 0.14643 & 0.18159 & 0.16092 & 0.17222 & 0.17139 & 0.35044 & 0.34383 & 0.34818 & 0.35760 & 0.36804\\
    \hline
    8 & 0.13860 & 0.17605 & 0.15938 & 0.17078 & 0.16192 & 0.34086 & 0.32983 & 0.35769 & 0.34430 & 0.34780\\
    \hline
    9 & 0.15047 & 0.16926 & 0.15153 & 0.17005 & 0.15266 & 0.33137 & 0.31749 & 0.34690 & 0.33840 & 0.34028\\
    \hline
    10 & 0.14734 & 0.14825 & 0.14585 & 0.15001 & 0.15459 & 0.3170 & 0.32975 & 0.32772 & 0.34529 & 0.32670\\
    \hline
    11 & 0.14227 & 0.14135 & 0.14326 & 0.13797 & 0.14507 & 0.31917 & 0.31851 & 0.32689 & 0.33667 & 0.32043\\
    \hline
    12 & 0.13818 & 0.14448 & 0.14703 & 0.13589 & 0.14142 & 0.32137 & 0.32451 & 0.31843 & 0.32556 & 0.34014\\
    \hline
    13 & 0.16111 & 0.13641 & 0.14893 & 0.14208 & 0.13808 & 0.31224 & 0.31219 & 0.31380 & 0.31515 & 0.31159\\
    \hline
    14 & 0.13894 & 0.13111 & 0.14094 & 0.14199 & 0.14192 & 0.30883 & 0.30496 & 0.30578 & 0.30726 & 0.30282\\
    \hline
    15 & 0.15106 & 0.14271 & 0.14601 & 0.13584 & 0.12413 & 0.29405 & 0.31198 & 0.30786 & 0.31167 & 0.30100\\
    \hline
  \end{tabular}}

\end{table}

\begin{table}[htbp]
\centering
\caption{EMD between the true and the estimated PMFs by PRIVIC on the San Francisco locations.}~\label{table:UtilityData_SF}
\resizebox{\columnwidth}{!}{%
  \begin{tabular}{cccccc|ccccc}
    \hline
    \multirow{2}{*}{$N$} & 
      \multicolumn{5}{c|}{$\beta=1$} & 
      \multicolumn{5}{c}{$\beta=0.5$} \\
    & Round 1 & Round 2 & Round 3 & Round 4  & Round 5 & Round 1 & Round 2 & Round 3 & Round 4 & Round 5 \\
    \hline
    1 & 7.37595 & 7.37595 & 7.37595 & 7.37595 & 7.37595 & 7.37595 & 7.37595 & 7.37595 & 7.37595 & 7.37595\\
    \hline
    2 & 0.37229 & 0.37038 & 0.36784 & 0.36949 & 0.36816 & 0.79621 & 0.80474 & 0.80219 & 0.80670 & 0.79940\\
    \hline
    3 & 0.29828 & 0.298370 & 0.30017 & 0.29784 & 0.29859 & 0.64931 & 0.66292 & 0.65362 & 0.65950 & 0.65285\\
    \hline
    4 & 0.26091 & 0.26029 & 0.26231 & 0.26180 & 0.26518 & 0.56896 & 0.57378 & 0.57338 & 0.57125 & 0.56618\\
    \hline
    5 & 0.23472 & 0.23419 & 0.23337 & 0.23740 & 0.23897 & 0.51672 & 0.51710 & 0.51735 & 0.51880 & 0.51138\\
    \hline
    6 & 0.21367 & 0.21432 & 0.21537 & 0.21777 & 0.21881 & 0.48194 & 0.48299 & 0.47992 & 0.48267 & 0.47791\\
    \hline
    7 & 0.19612 & 0.19761 & 0.19904 & 0.20120 & 0.20067 & 0.45531 & 0.45861 & 0.45122 & 0.45732 & 0.45450\\
    \hline
    8 & 0.18244 & 0.18412 & 0.18741 & 0.18724 & 0.18674 & 0.43588 & 0.43745 & 0.43558 & 0.43704 & 0.43584\\
    \hline
  \end{tabular}%
  }
\end{table}

\end{document}